\documentclass[conference]{IEEEtran}
\ifCLASSINFOpdf
\else
\fi
\usepackage{amsthm,amssymb,amsfonts}
\usepackage{algorithm}
\usepackage{tikz,lipsum}
\usepackage[most]{tcolorbox}

\newtcolorbox[auto counter]{Box1}[2][]{
                lower separated=false,
                colback=white,
                breakable,pad at break*=0mm,vfill before first,
colframe=white!20!gray,fonttitle=\bfseries,
colbacktitle=white!10!gray,enhanced,
attach boxed title to top left={xshift=1cm,
        yshift=-2mm},
title=#2,#1}
\tcbuselibrary{skins,breakable}
\usepackage{graphicx}  % For including images
\usepackage{diagbox}
\usepackage{comment}
\usepackage{hyperref}
\usepackage{amssymb}
\usepackage{threeparttable}
\usepackage{multirow}
\usepackage[noend]{algpseudocode}
\usepackage{algpseudocode}
\newtheorem{theorem}{Theorem}
\usepackage{longtable}
\newtheorem{definition}{Definition}
\newtheorem{lemma}{Lemma}
\newtheorem{corollary}{Corollary}
\usepackage[functional,reproduced]{ndssbadges}
\begin{document}
\title{A Two-Layer Blockchain Sharding\\ Protocol Leveraging Safety and Liveness\\ for Enhanced Performance}
\makeatletter
\newcommand{\linebreakand}{%
  \end{@IEEEauthorhalign}
  \hfill\mbox{}\par
  \mbox{}\hfill\begin{@IEEEauthorhalign}
}
\makeatother

\author{
  \IEEEauthorblockN{Yibin Xu}
  \IEEEauthorblockA{
    Department of Computer Science \\
    University of Copenhagen\\
    yx@di.ku.dk
  }
  \and
  \IEEEauthorblockN{Jingyi Zheng}
  \IEEEauthorblockA{
    Department of Computer Science\\
    University of Copenhagen\\
    jrb385@alumni.ku.dk
  }
  \and
  \IEEEauthorblockN{Boris D{\"u}dder}
  \IEEEauthorblockA{
    Department of Computer Science\\
    University of Copenhagen\\
    boris.d@di.ku.dk
  }
  \linebreakand
  \IEEEauthorblockN{Tijs Slaats}
  \IEEEauthorblockA{
    Department of Computer Science\\
    University of Copenhagen\\
    slaats@di.ku.dk
  }
  \and
  \IEEEauthorblockN{Yongluan Zhou}
  \IEEEauthorblockA{
    Department of Computer Science\\
    University of Copenhagen\\
    zhou@di.ku.dk
  }
}

\IEEEoverridecommandlockouts
\makeatletter\def\@IEEEpubidpullup{6.5\baselineskip}\makeatother
\IEEEpubid{\parbox{\columnwidth}{
    Network and Distributed System Security (NDSS) Symposium 2024\\
    26 February - 1 March 2024, San Diego, CA, USA\\
    ISBN 1-891562-93-2\\
    https://dx.doi.org/10.14722/ndss.2024.24006\\
    www.ndss-symposium.org
}
\hspace{\columnsep}\makebox[\columnwidth]{}}

% make the title area  
\maketitle

\begin{abstract}

Sharding is a critical technique that enhances the scalability of blockchain technology. However, existing protocols often assume adversarial nodes in a general term without considering the different types of attacks, which limits transaction throughput at runtime because attacks on liveness could be mitigated. There have been attempts to increase transaction throughput by separately handling the attacks; however, they have security vulnerabilities. This paper introduces Reticulum, a novel sharding protocol that overcomes these limitations and achieves enhanced scalability in a blockchain network without security vulnerabilities.

Reticulum employs a two-phase design that dynamically adjusts transaction throughput based on runtime adversarial attacks on either or both liveness and safety. It consists of `control' and `process' shards in two layers corresponding to the two phases. Process shards are subsets of control shards, with each process shard expected to contain at least one honest node with high confidence. Conversely, control shards are expected to have a majority of honest nodes with high confidence. Reticulum leverages unanimous voting in the first phase to involve fewer nodes in accepting/rejecting a block, allowing more parallel process shards. The control shard finalizes the decision made in the first phase and serves as a lifeline to resolve disputes when they surface.

Experiments demonstrate that the unique design of Reticulum empowers high transaction throughput and robustness in the face of different types of attacks in the network, making it superior to existing sharding protocols for blockchain networks.

%In the first phase, transactions are written to blocks, which are then voted on by nodes in the corresponding process shards. Blocks that receive unanimous acceptance verdicts are accepted. In the second phase, blocks that do not obtain unanimous verdicts are collected and voted on by the control shards. A block is accepted if the majority of nodes vote to accept it. The opponents and silent voters of the first phase will be eliminated. In summary, 

\end{abstract}
\begin{comment}
\begin{IEEEkeywords}
Blockchain,
Sharding,
Scalability,
Adversary nodes,
Adaptive sharding protocol
\end{IEEEkeywords}
\end{comment}
\section{Introduction}
\label{introduction}
Blockchain sharding~\cite{luu2016secure,zamani2018rapidchain,kokoris2018omniledger,dang2019towards,zheng2021meepo,david2022Gearbox,hong2021pyramid,lewenberg2015inclusive,wang2019monoxide,huang2020repchain,khacef2022trade,kiayias2017ouroboros,10.1145/3548606.3563506,xu2023two} is a method that aims to improve the scalability of a vote-based blockchain system by randomly dividing the network into smaller divisions, called shards. The idea is to increase parallelism and reduce the overhead in the consensus process of each shard, thereby increasing efficiency. However, there is a trade-off between parallelism and security in sharding. While making smaller and more shards can increase parallelism, they are more vulnerable to corruption.

The security of a shard has two key properties: liveness and safety. Liveness concerns the shard's capability to ultimately achieve a consensus on the sequence of output messages, whereas safety pertains to the accuracy and exclusivity of that consensus. For example, if we have $M=10$ nodes in a synchronous communication shard, a decision is made only if at least seven nodes vote in favor. Then, $S=6/M$ is the maximum ratio of adversarial population that can be tolerated. Otherwise, an adversarial decision may be reached. However, if at least a $1-S$ ratio of nodes in the shard always disagree or keep silent with any proposals, the shard can never reach any decisions, and the liveness is compromised. The safety threshold ($S$) and liveness threshold ($L$) represent the maximum ratio of adversarial participants in a shard under which safety and liveness are guaranteed.

However, existing sharding solutions~\cite{zamani2018rapidchain} guarantee security under worst-case adversarial conditions ($L= S<50\%$) in the synchronous network and assume that the adversary has equal interests in attacking both safety and liveness simultaneously. It limits the number of parallel shards and the performance gains of sharding when the nodes acting adversarial at runtime is lower than the worst-case. 

Recent proposals~\cite{david2022Gearbox,xu2020n,xu2020flexible} aim to increase the number of parallel shards by increasing $S$ and decreasing $L$ of each shard. This is based on the fact that $S<1-L$ (synchronous model) or $S<1-2L$ (partially synchronous model) holds because given $L$ percent of nodes that attack liveness, a block will only be accepted if at least $1-L$ or $1-2L$ percent of nodes respectively have voted in favor of it. Therefore, the system cannot tolerate more than $1-L$ or $1-2L$ percent of nodes, respectively, that attack safety. 
Based on this theory, one can generate smaller shards with higher $S$ (and lower $L$) to tolerate the same number of adversaries globally. Note that $S\geq L$ is maintained in all cases.
Fig.~\ref{fig:new1} illustrates changing $L$ and $S$ by adapting the shard size proposed in~\cite{david2022Gearbox,xu2020flexible}. This design allows using small shards when the adversarial population ratio that intends to attack liveness is low, so more shards run in parallel. When this ratio is higher at runtime, the shards can be respawned with a larger shard size, resulting in an increased $L$, and a decreased $S$~\cite{david2022Gearbox,xu2020flexible}. The existing shard resizing approaches either reform the shards from scratch, changing the shard memberships of all the nodes~\cite{xu2020flexible}, or perform local shard adjustments, resulting in overlapping shards~\cite{david2022Gearbox} as shown in Fig.~\ref{fig:new1}.

\begin{figure}[h!]
    \centering
    \includegraphics[width=0.5\textwidth]{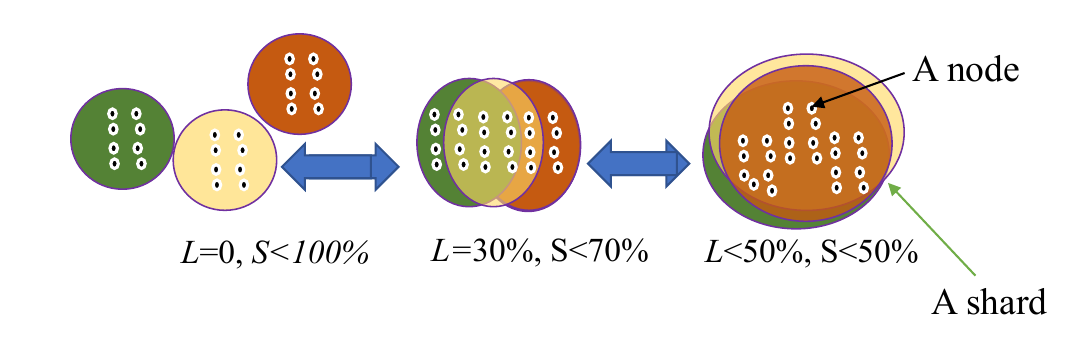}
    \caption{The typical structure of a sharding approach leveraging liveness and safety thresholds. The size of shards increases with the liveness threshold $L$ increase. There are three shards in the system. The increased size either reduces the number of shards or causes shards to overlap.}
    \label{fig:new1}
\end{figure}

Despite the potential performance gains associated with leveraging liveness and safety, several key limitations must be addressed. \textit{First}, the approaches of adapting shard sizes and memberships, as mentioned earlier, can result in frequent and costly shard respawning. Additionally, the presence of overlapping shards when the adversarial population fluctuates leads to additional costs. 
\textit{Secondly}, our analysis in Sec.~\ref{breaches} uncovers security vulnerabilities present in all existing designs. In particular, synchronous shards with $S \geq \frac{1}{2}$ or partially synchronous shards with $S \geq \frac{1}{3}$ are unable to achieve consensus without the possibility of equivocation. These vulnerabilities undermine the effectiveness of the approaches.
\textit{Thirdly},  all approaches so far have no design to mitigate the liveness attacks. Such attacks can be detected and evidenced in runtime, allowing for appropriate punitive actions when liveness is restored while ensuring safety is maintained at all times. 
For instance, if there are methods to expel a node when this node disagrees with correct proposals or frequently remains silent, such adversarial nodes cannot again attack liveness in the future after being expelled.

\subsection{Contribution}
The contributions of this paper are as follows:

\noindent \textbf{\textit{(1) We propose Reticulum, the first protocol leveraging liveness and safety but does not suffer from security breaches. Reticulum can inhibit adversarial behaviours.}} These are achieved without needing to respawn new shards in runtime or to overlap shard memberships, which can bring unnecessary overhead.

\noindent \textbf{\textit{(2) We comprehensively analyze Reticulum's performance with a comparison with state-of-the-art approaches.}} We evaluate Reticulum's performance empirically by both simulation and experiments. We compare Reticulum with two state-of-the-art approaches: Gearbox~\cite{david2022Gearbox} and Rapidchain~\cite{zamani2018rapidchain}. Our analysis shows that Reticulum significantly outperforms both approaches regarding transaction throughput and storage requirements. 

%\noindent \textbf{\textit{(3) We have implemented and open-sourced a prototype of Reticulum}}. The lack of available implementations of blockchain sharding protocols has created a significant challenge for researchers attempting to compare different protocols through experimentation. Our open-source prototype \cite{githubGitHubJingyiZhengReticulum} is an important step toward addressing this issue and will provide researchers with a valuable resource for studying blockchain sharding. 

\section{Overview}
\label{1.2}

To address the shortcomings of existing approaches, we propose Reticulum, a blockchain sharding approach designed for a synchronous environment. For simplification, Reticulum uses a static sharding scheme, i.e., no new members are added in runtime and we do not respawn any shards. But in reality, the same as Rapidchain~\cite{zamani2018rapidchain} and Omniledger~\cite{kokoris2018omniledger}, it may allow a system reboot with fresh nodes added and some nodes removed every several days using the same design. It may also use the same design as Rapidchain to swap nodes of different shards after several epochs to avoid adversaries (slowly) enlarging the corruption population in a shard (assuming they are capable of doing so). Reticulum adopts two layers of ledgers to provide resilience. Fig.~\ref{fig:struct} illustrates its structure. 

\begin{figure}[h!]
    \centering
    \includegraphics[width=0.5\textwidth]{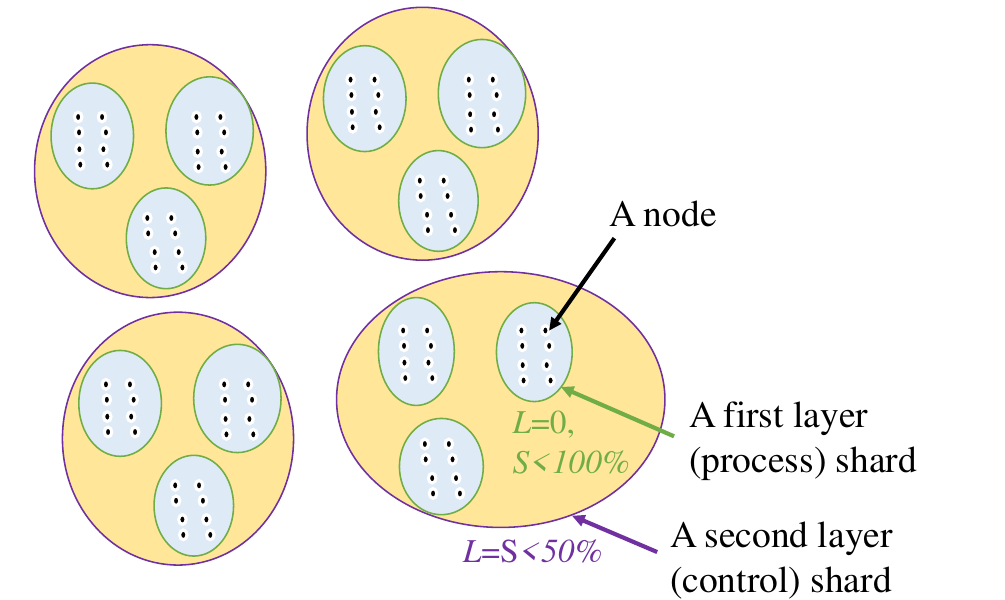}
    \caption{The structure overview of Reticulum protocol. Every node is, at the same time, in only one process shard ($L=0, S<100\%$) and one control shard ($(L=S)<50\%$) that governs this process shard.}
    \label{fig:struct}
\end{figure}

Reticulum employs a two-phase-voting mechanism in which a blockchain epoch is divided into two phases. Each node is assigned to one and only one process shard, which consists of the first layer of ledgers. Each process shard is assigned to and governed by one and only one control shard, which consists of the second layer of ledgers.
The shard memberships do not change at runtime and they do not overlap. 
Moreover, all nodes are connected to a public communication chain, which is used for communicating metadata.

In the first phase, the members of a process shard generate and vote for the acceptance of a process block containing transactions associated with the data stored. All process shards have $S<100\%$, which means that a block will only be accepted with a unanimous vote. 
On the other hand, the process shards have $L=0\%$, which means that 
one adversary would be sufficient to hinder the progress.  

In order to prevent progress hindrance caused by an acceptable block failing to receive a unanimous verdict in a process shard, a control shard is used in the second phase to reach a majority verdict on the block that did not pass in the first phase. Each control shard has a threshold of $L, S<50\%$, which ensures that it provides comparable worst-case security to other state-of-the-art blockchain sharding protocols. The votes in a \textit{process} shard are broadcast to all nodes in the corresponding \textit{control} shard (sized $N_c$) using a Byzantine reliable broadcast protocol called $(\Delta+\delta)$-BRB~\cite{abraham2021good} with $f\leq \lfloor (N_c-1)/2\rfloor$, this enables them to align that if a process block needs to be handled in the control shard. 

We design that when a process block cannot be accepted in a process shard, but later be accepted by the control shard, all nodes in the process shard which voted for rejection will be expelled from the system. If a node remained silent in voting, it is marked as ``violated node'' in the current voting. A node can only be marked as ``violated node'' once in every $\tau$ votings for the process blocks, where $\tau$ is a pre-defined liveness threshold. If node $i$ is marked twice, it is also expelled from the system. 

It is important to note that the cost of the second voting phase will depend on the number of failed process blocks in the first phase.  As the adversarial population decreases (or increases), fewer (or more) process blocks will likely fail the first voting phase, which will result in a lower (or higher) cost for the second phase and higher (or lower) transaction throughput. When no adversaries are attacking liveness, Reticulum will achieve its highest transaction throughput because it has no failed process block to vote for in the second phase. In this way, Reticulum is resilient to the dynamic decrease (or increase) of the adversarial population with higher (or lower) transaction throughput.

\section{Problem domain}
This section defines the network and threat models for our sharding protocol, the overview of Reticulum, and the objectives of Reticulum. 

\subsection{System model}\label{s2a}
\noindent \textbf{\textit{Network model: standard synchronous network model.}} Similar to Rapidchain~\cite{zamani2018rapidchain}, our network model consists of a peer-to-peer network with $N$ nodes that use a Sybil-resistant identity generation mechanism to establish their identities (i.e., public/private keys). This mechanism requires nodes to attach a Proof-of-Stake (PoS) of a threshold weight that all other honest nodes can verify. We assume that all messages sent in the network are authenticated with the sender's private key. These messages are propagated through a synchronous gossip protocol~\cite{karp2000randomized}, which guarantees that the message delays are at most a known upper bound $\Delta$. Note that this protocol does not necessarily preserve the order of the messages. 

It is worth noting that Gearbox~\cite{david2022Gearbox} claims to utilize a mixed network model based on $L$ and $S$, enabling the execution of a partially synchronous model until $L=25\%$ and $S=49\%$. However, as will be demonstrated in Sec.~\ref{breaches}, in practice, the decision to accept or reject blocks of a shard in Gearbox can only be safely reached after finalization. Therefore, it is the confirmation of the heartbeat within a block of the ``control chain'' that signifies the completion of the BFT process of each shard. Since the blocks of the ``control chain'' also incorporate heartbeats from other shards, some of which may operate with $L>25\%$ at the time, in many cases, the block is generated only after a known time-bound, which implies a synchronous network model for Gearbox.

\noindent \textbf{\textit{Threat model: adaptive but upper-bounded adversaries with  $\tau$ liveness guarantee.}} In our threat model, we consider the possibility of nodes disconnecting from the network due to internal failures or network jitter, as well as a probabilistic polynomial-time Byzantine adversary that can corrupt up to $f \leq \lfloor (N-1)/3\rfloor $ nodes at any given time. These adversarial nodes may collude with each other and deviate from the protocol in arbitrary ways, such as by sending invalid or inconsistent messages or remaining silent. Note that a non-sharded synchronous communication environment can tolerate $f \leq \lfloor (N-1)/2 \rfloor$. However, for a sharded network with shard size $M$, it is impossible to withstand $f \leq \lfloor (M-1)/2 \rfloor$ adversary (the mere majority honesty), while allowing $f \leq \lfloor (N-1)/2 \rfloor$ globally. This is because the failure probability (more than the threshold adversary nodes being assigned to the shard) cannot be decreased to a trivial number with a random shard assignment for nodes. 

Rapidchain~\cite{zamani2018rapidchain} and Gearbox~\cite{david2022Gearbox} also require that, at any given moment, at least 2/3 of the computational resources belong to uncorrupted online participants. Since the runtime adversarial population is unknown to the uncorrupted participants, achieving \textit{deterministic} consensus requires constant liveness for the same set of $\lfloor \frac{2(N-1)}{3}\rfloor+1$ honest nodes, as shown in Thm.~\ref{tt1}. \\

\begin{theorem}
For guaranteeing \textbf{constant} liveness and safety in the overall system, a blockchain sharding protocol with a safety threshold $S$ and a liveness threshold $L$ (both less than 50\%) necessitates the presence of at least $\left\lfloor \frac{2(N-1)}{3} \right\rfloor + 1$  uncorrupted participants who remain active in the system \textbf{at all times.} \label{tt1}
\end{theorem}
\begin{proof}
Let us assume that the adversary has control over $f = \left\lfloor \frac{N-1}{3} \right\rfloor$ nodes across the entire network, and these nodes persistently act in an adversarial manner. Under this scenario, all other nodes must maintain their integrity to satisfy the liveness guarantees and must not leave the system. This is essential to avoid the following situation: Only one uncorrupted node among the $N-\left\lfloor \frac{N-1}{3} \right\rfloor$ uncorrupted participants leaves the sharding protocol during a specific epoch, and it happens that this node belongs to a shard containing $f = \left\lfloor \frac{M-1}{2} \right\rfloor$ nodes. Then the shard becomes unable to deterministic accept an uncorrupted block for this specific epoch. Consequently, both liveness and safety are lost. The loss of liveness and safety in a single shard compromises the constant liveness and safety guarantee of the overall system.
\end{proof}

We use the $\tau$ liveness guarantee for the adversarial nodes, which is defined in Def.~\ref{Theorem}. If a node violates $\tau$ liveness, it will be expelled from the system.\\
\begin{definition}$\tau$ liveness guarantee refers to a blockchain sharding system that mandates that at least $\left\lfloor \frac{2(N-1)}{3} \right\rfloor + 1$ uncorrupted participants constantly act uncorrupted and remain alive, but allowing any other participant to not participate in the system once in every $\tau$ epochs.\label {Theorem}
\end{definition}

\subsection{System overview}
Reticulum has a determined number of nodes $N$. Each node $n_i$ of these $N$ nodes is a member of the public communication chain (PC) and is a member of exactly one process shard and one control shard. 
We elaborate on our designs below:\\

\noindent \textbf{\textit{Public communication chain.}} The public communication chain (PC) coordinates the sharding protocol and stores metadata about the shards. It is a totally-ordered broadcast with persistence and a timestamp guarantee, and all nodes in the system execute it. The PC stores only metadata, which means that its size is independent of the contents of the shards. \\
\\\noindent \textbf{\textit{Shards.}} Shards are modeled as ledger functionalities that are parameterized by the size of the shard and the type (process/control) of the shard. There are $\Lambda$ control shards, each containing at least $N_c = \lfloor N/\Lambda \rfloor$ nodes, and $\beta$ process shards, each containing at least $N_p = \lfloor N/\beta \rfloor$ nodes. Each process shard is governed by only one control shard. Each control shard supervises at least $\lfloor N_c/N_p\rfloor$ process shard. Each process shard governs a non-overlapping set of keys, and the related transactions should be sent to the corresponding process shards. Cross-shard transactions that interact with multiple shards are used if different process shards govern the keys involved in the transactions. Each process or control shard runs an ordinary synchronous BFT protocol, where a leader node is randomly selected to propose a block followed by one round of voting. A process block is accepted in a process shard if all nodes vote to accept it. A control block is accepted in a control shard if the majority of nodes vote to accept it. In normal cases, the control shard does not synchronize the blocks of the process shards. \textit{However, they synchronize the votes for the blocks of each process shard under their governance.} If a block in a process shard does not obtain unanimous votes within a time-bound ($T_1$), it is synced by all nodes in the corresponding control shard in another time-bound($T_2$). The control shard votes on them to finalize the verdicts.\\

\noindent \textbf{\textit{State-block-state structure.}}  The state of the process shard is updated with each process block, and a new state is formed in each epoch. Rather than linking blocks to preceding blocks, they are linked to previous states. If a block does not pass unanimous voting, only the latest state and the rejected block need to be synced by nodes in the control shard, instead of the entire blockchain of the process shard. This greatly reduces the amount of data that needs to be synced, improving the system efficiency.\\

\noindent \textbf{\textit{Bootstrapping.}} The system is initialized with a predetermined list of $N$ participating nodes, which are ranked using a global random number generated by a distributed random number generator~\cite{das2022practical}. Nodes are assigned to different shards after bootstrapping according to rank. \\

\noindent \textbf{\textit{Simple Cross-shard communication.}} Cross-shard communication, which involves the exchange of information and data between different shards in a sharding system, has been extensively researched in the literature~\cite{kokoris2018omniledger,zamyatin2021sok,avarikioti2019divide,wang2019sok,huang2022brokerchain,wang2021low}. However, implementing a reliable cross-shard communication protocol in a sharding system that potentially have liveness issues be challenging. Gearbox, for instance, customized the existing protocol Atomix~\cite{kokoris2018omniledger} to finalize messages on the shards to address this issue, which further complicated the algorithm. In contrast, Reticulum eliminates the need for specialized cross-shard protocols altogether. Since we do not adjust shard membership, signatures contained in cross-shard transactions can be trusted without considering if they are signed by the correct shard members at the time. As a result, we can use ordinary cross-shard protocols for cross-control-shard transactions without worrying about the complications associated with deadlocked shards.

\subsection{Objectives} 
 Our protocol involves processing a set of transactions submitted to our protocol by external users. Similar to those in other blockchain systems, these transactions consist of inputs and outputs referencing other transactions and are accompanied by a signature for validity. The set of transactions is divided into $\beta=\lfloor N/N_p \rfloor$ blocks proposed by disjoint process shards each sized $N_p$, with $y_{\{i,j\}}$ representing the $j$-th transaction in the block of $i$-th process shard. There are $\Lambda = \lfloor N / N_c \rfloor$ control shards, each sized $N_c$. The $i$-th process block is governed by $\lfloor i / N_c \rfloor$-th control shard. All nodes within the same process shard have access to an external function $g1$ that determines the validity of a transaction, outputting 0 or 1 accordingly. All nodes in the same control shard have access to an external function $g2$ that determines the validity of a process block that did not pass the unanimous voting in the first phase, outputting 0 or 1 accordingly. $Block(i)$ is the block of process shard $i$ of the current epoch. If $Block(i)$ fails to obtain unanimous acceptance in the first phase and $g2(Block(i)))=1$, the leader node of $\lfloor i / N_c \rfloor$-th control shard in the current epoch should embed $Block(i)$ to the control block of current epoch and mark it as ``accepted''.
  $z_\{p,q\}$ represents the $q$-th process block embedded in the $p$-th control block that is marked as ``accepted''. The protocol $\Pi$ aims to output a set $(Y, Z)$ containing $\beta$ disjoint process blocks $Y_{i} = \{y_{\{i,j\}}\}$ where $j \in \{1\dots |Y_{i}|\}$ and $\Lambda$ control blocks $Z_{p} = \{z_{\{p,q\}}\}$ where $q \in \{1\dots |Z_{p}|\}$, satisfying the following conditions:
  
\noindent\begin{itemize}
\item Agreement: For any $i \in \{1\dots\beta\}$, at least one honest node agree on $Y_{i}$ with a probability of greater than $1 - 2^{-\sigma}$, where $\sigma$ is the security parameter. In the meantime, $\Omega(\log N)$ honest nodes agree on $Z_{p}$ with a probability of greater than  $1 - 2^{-\sigma}$. \\
\item Validity: For each $i \in \{1 \dots \beta\}$ and $j \in \{1\dots|Y_{i}|\}$, $g1(y_{\{i,j\}}) = 1$.  For each $p \in \{1 \dots \Lambda\}$ and $q \in \{1\dots|Z_{p}|\}$, $g2(y_{\{p,q\}}) = 1$. \\
\item Termination: Each shard $i \in \{1\dots\beta\}$ decide (approved unanimously or not) on $Y_i$ within $T_1$. Each shard in $p \in \{1\dots\Lambda\}$ decide on $Z_p$ within $T_2$. 
\end{itemize}

\noindent The design objectives of Reticulum are to achieve the following while satisfying the above conditions.
\noindent\begin{itemize}
\item Scalability: The number of control shards $\Lambda$ and process shards $\beta$ in Reticulum can grow with the increase of the total number of nodes in the network, denoted by $N$, without affecting the aforementioned conditions. This design ensures that Reticulum can handle a growing network and maintain high throughput. Furthermore, as the number of honest nodes in the network increases, Reticulum's throughput also increases. 
\item Efficiency: Reticulum is designed to minimize per-node  computation, and storage complexity. The per-node computation complexity is $O(N)$, where $N$ is the total number of nodes in the network. The per-node storage complexity is $O(OT/k)$, where $OT$ is the total number of transactions, and $k$ is the average number of process shards
\end{itemize}

In summary, our objective is to design a consensus protocol that achieves agreement among a sufficient number of honest nodes, ensuring the validity of the agreed-upon transactions, scales with the size of the network and the runtime adversary population, and has efficient communication, computation, and storage cost.

\section{Reticulum protocol}
This section will provide a detailed description of the Reticulum protocol. Reticulum comprises three main components: bootstrapping, two-layer-shard consensus, and cross-shard transactions. 

In Reticulum, each node participates in the public communication chain (PC) and exactly one process shard, along with its corresponding control shard. This means that each node maintains two additional chains besides PC: the \emph{process-shard-chain} specific to the process shard and the \emph{control-shard-chain} specific to the control shard.
A block proposed for the process-shard-chain is called a \emph{process block}, while that for the control-shard-chain is called a \emph{control block}. A process block includes a collection of transactions, and the control block includes the hashes of process blocks and their vote signatures.

\subsection{Bootstrapping} The Bootstrapping phase is used to assign nodes to shards randomly and unbiasedly. We first share a predetermined list of $N$ nodes with all participants. Then, a random sequence $C$, where $|C|=N$, is generated collectively utilizing a random beacon~\cite{das2022practical}. The random beacon must be accessible to all participants, verifiable by all parties, and should not be controlled by any entity. 

We use a $getShardIndex$ function to determine a node's shard membership. This function takes a node $Node_j$ as input, retrieves the index $C_{index}$ of $Node_j$ in the sequence $C$, and then calculates the indices $index_1$ and $index_2$ for the node's process and control shards, respectively. The output is a tuple $(index_1, index_2)$. Note that bootstrapping only runs at the beginning of the protocol or when the protocol re-starts with new node membership.
The pseudo-code for  $getShardIndex$ is given in Alg.~\ref{euclid}. 
\begin{algorithm}
    \caption{getShardIndex}\label{euclid}
    \begin{algorithmic}[1]
    \State $N_c$ and $N_p$ are the numbers of nodes inside a control and process shard, respectively, $N_c|N_p$. A sequence $C$ is generated using a distributed random beacon from $\{Node_1,Node_2, \dots, Node_N\}$. $|C|=N$. When $N\nmid N_c$, the number of nodes inside the last control shard and the last process shard may exceed $N_c$ and $N_p$.
    \Function{getShardIndex}{$Node_j$}
    \State $C_{index}\gets$ The index of $Node_j$ in the sequence $C$ starting from 0
    \State $index_1,index_2\gets $ $\lfloor C_{index} / N_p\rfloor,\lfloor index_1/( N_c/N_p) \rfloor$
    \If {$C_{index}\geq \lfloor N/N_c\rfloor\times N_c$}
        \State $index_2\gets  index_2-1$ \Comment{Add nodes to the last control shard but not a new control shard}
                \If {$N\nmid N_p$}
                \State $index_1\gets \min \left (\lfloor (C_{index}/N_p)\rfloor, \lfloor (N/N_p)\rfloor-1 \right )$\Comment{Add nodes to the last process shard but not a new process shard}
      
    \EndIf
    \EndIf

    \State \textbf{return} $(index_1,index_2)$ \Comment{$Node_j$ belongs to $ps_{index_1}$ and $cs_{index_2}$}
    \EndFunction
    \end{algorithmic}
    
\end{algorithm}
\begin{lemma}
The bootstrapping phase guarantees the safety property of shard membership assignment, ensuring that no node is assigned to multiple shards.
\end{lemma}
\begin{proof}
    Please see the full proof in Appx.~\ref{a1}.
\end{proof}

\subsection{The first phase}
After nodes are randomly assigned to the process shard and the control shard during bootstrapping, Reticulum enters the consensus stage. An epoch for the consensus stage consists of two phases, with the first phase to be completed within a time-bound denoted as $T_1$, and the second phase to be completed within a time-bound denoted as $T_2$. While $T_1$ is set as a constant number, $T_2$ depends on the number of process shards and their verdicts in the first phase. We will discuss this in detail in Sec.~\ref{4.2.2}. 

The first phase is used for the process shards to reach a consensus on the process blocks. After bootstrapping, there are at least $N_p$ nodes in any process shard $ps_i$, $i \in [0,\beta]$. At the beginning of the first phase, each $ps_i$ tries to generate and reach consensus on a block using a standard binary vote-based synchronous consensus protocol~\cite{DBLP:journals/corr/0001NAD17} within the time-bound $T_1$. Each process shard requires a unanimous vote to decide on a block. In other words, as long as there is one adversarial node in $ps_i$, it cannot decide on a block and cannot evolve into a new state. Therefore the process blocks have $L=0$. Apparently, the consensus protocol in the first phase only yields two possible outcomes:
\begin{itemize}
    \item\textbf{A block is unanimously accepted:} This implies that there is a probability greater than $1 - 2^{-\sigma}$ that no nodes act adversarial in this process shard. 
    \item \textbf{A block has not received unanimous voting:} Because of the synchronous communication assumption, all honest nodes vote and receive each other's vote within $T_1$. Because we are using a binary vote-based consensus protocol, the honest nodes will decide on the same verdict (either accept or reject) of a block. When a block has not received unanimous voting, it implies that there are adversarial nodes in the shard.
\end{itemize}
Note that the votes should be broadcasted to all nodes within the same control shard governing this process shard using the Byzantine reliable broadcast protocol $(\Delta+\delta)$-BRB~\cite{abraham2021good} with $f\leq \lfloor (N_c-1)/2\rfloor$ to allow nodes in the other process shards to know the voting output and the deterministic set of votes. The pseudo-code for the first phase of the two-layer consensus is given in Alg.~\ref{consensus-protocol-first-stage}.

\begin{algorithm}
\caption{First phase of Two-layer consensus}\label{consensus-protocol-first-stage}
\begin{algorithmic}[1]
\Procedure{FirstPhase}{$node$}\Comment{$node$ run this procedure}
\State $votes \gets$ empty set of votes
\If{$\text{this node is the leader of its process shard}$} 
\State $block \gets \text{GENERATEBLOCK}(node)$
\State $\text{BROADCASTBLOCK}(block, node.ps)$ \Comment{broadcast the block to all nodes in the same process shard}
\EndIf
\If{$\text{this node is not the leader of its process shard}$} 
\State $block \gets \text{RECEIVEBLOCK}()$
\EndIf

\State $vote \gets \text{GENERATEVOTE}(node, block)$ \Comment{verify and vote for the block}
\State $\text{BROADCASTVOTE}(vote, node.cs)$ \Comment{broadcast the vote to all nodes in the same control shard using a BRB protocol}
\While{$\text{in $T_1$ and CONSENSUS(votes, node.ps)=FAILED}$} \Comment{wait for votes}% until $T_1$}
\State $votes \gets votes \cup {\text{RECEIVEVOTE}()}$ \Comment{add the vote to the set of votes}
\EndWhile

\State $result \gets \text{CONSENSUS}(votes, node.ps)$ 
\State \textbf{return} $result$ 
\EndProcedure

\Function{broadcastBlock}{$block, ps$}
\For{$n \in ps$}
\State send $block$ to node $n$
\EndFor
\EndFunction

\Function{broadcastVote}{$vote, cs$}
\For{$n \in cs$}
\State send $vote$ to node $n$
\EndFor
\EndFunction

\Function{consensus}{$votes,ps$}
\State $count_1 \gets 0$
\For{$v \in votes$}
\If{$v.value = approve$ and $v.ps=ps$}
\State $count_1 \gets count_1 + 1$
\EndIf
\EndFor
\If{Unanimously approved}
\State \textbf{return} $ACCEPTED$
\Else
\State \textbf{return} $FAILED$
\EndIf
\EndFunction
\end{algorithmic}
\end{algorithm}

\begin{theorem}
    In Reticulum, nodes can reach a consensus for whether a node is identified as the adversarial node when the votes of the first phase are broadcasted using Byzantine reliable broadcast protocol. \label{lemma2}

\end{theorem}
\begin{proof}
Blockchain consensus ensures the validation of the consensus for the sequence of transactions, rather than the correctness of individual transactions. The correctness of transactions can be verified by referring to information in previous blocks of the blockchain. Therefore, a node will only vote to reject blocks if they contain violations such as double-spending or over-spending transactions that contradict previous verdicts. The control blocks also include decisions for expelling nodes that did not participate in the voting for the process shards. To determine this in consensus, it is necessary for nodes to reach a consensus on whether a node voted for the process block within the time-bound. To support this design, when nodes vote in the first phase, the votes are broadcast to all nodes in the control shard (sized $N_c$) using a Byzantine reliable broadcast protocol called $(\Delta+\delta)$-BRB~\cite{abraham2021good} with $f\leq \lfloor (N_c-1)/2\rfloor$. 
It guarantees that, if a vote $V$ is broadcast to all nodes in a control shard of size $N_c$, all honest nodes in this control shard will terminate the broadcast protocol with the same value $V$. Hence, they can reach a consensus on if a node within the process shard voted or not. Consequently, when an adversarial node proposes/votes for a process block containing incorrect messages or its liveness violates $\tau$ liveness guarantees, it can be identified as an adversarial node in consensus.
\end{proof}

As shown in Thm.~\ref{lemma2}, Reticulum can safely identify nodes as adversarial in consensus. Reticulum will confiscate the Proof-of-Stake (PoS) of adversarial nodes and expel them from the system. We assume that most adversaries act rationally and do not wish to be marked as adversarial.
However, we demonstrate in Sec.~\ref{4.2.2} that Reticulum can also tolerate attacks when the adversary is willing to risk confiscation of their PoS.

\subsection{The second phase}
\label{4.2.2}
The second phase serves as a ``safety net'' assisting a final verdict for the blocks not passing the unanimous voting. The nodes in the process shards that fail to output unanimously approved blocks send the last state of the process shard and the failed block to all the other nodes in the same control shard.   

Similar to other vote-based consensus protocols, a leader node is selected to propose the control block in the second phase. The control block contains information about the process blocks of each process shard under its control. This information includes: (1) a boolean parameter indicating whether the process block passed unanimous voting within $T_1$, (2) if the block passed unanimous voting, the signatures from the voters are attached as proof, and (3) if the block did not pass unanimous voting, the hash of the block and a decision to accept or reject the block are attached. This decision is made by the creator of the control block, i.e., the current leader. Nodes can request the full content of the process blocks using the hash. Nodes in the control shard will verify the blocks received from the failed process shards and vote for the control block if they agree with the decision made by the current leader.

The completion of the second phase must be accomplished within $T_2$. Unlike $T_1$, the value of $T_2$ for a control shard is dynamically adjusted based on the count of succeeded process shards within that control shard, denoted as $N_s$. The calculation for $T_2$ can be expressed using the following equation:
\begin{equation}\label{eq1}
T_2 = \lambda \cdot (\lfloor N_c/N_p \rfloor - N_s+1)
\end{equation}
where $\lambda$ represents a pre-defined constant value, $N_p$ and $N_c$ are the size of process shards and control shards respectively. $T_2$ is set to ensure that the network bandwidth requirement for nodes in the second phase is similar to that in the first phase.    

After the aforementioned procedures, check if any node $i$ is considered the adversarial node. The PoS of the node $i$ is confiscated and the node $i$ is expelled from the system. Cor.~\ref{l1} shows that it is safe to operate with the remaining nodes. 

Note that if the majority of nodes in the control shard reject the control block, a new proposed block by a new leader will be voted within the same time-bound $T_2$. A new epoch is initiated only when a control block is accepted.

When a process block of epoch $X+1$ is rejected in an accepted control block, the corresponding process shard will not undergo any new state evolution in epoch $X+1$. Therefore, the state of that process shard in epoch $X+1$ will be the same as in epoch $X$.
\begin{figure}
    \centering
    \includegraphics[width=0.5\textwidth]{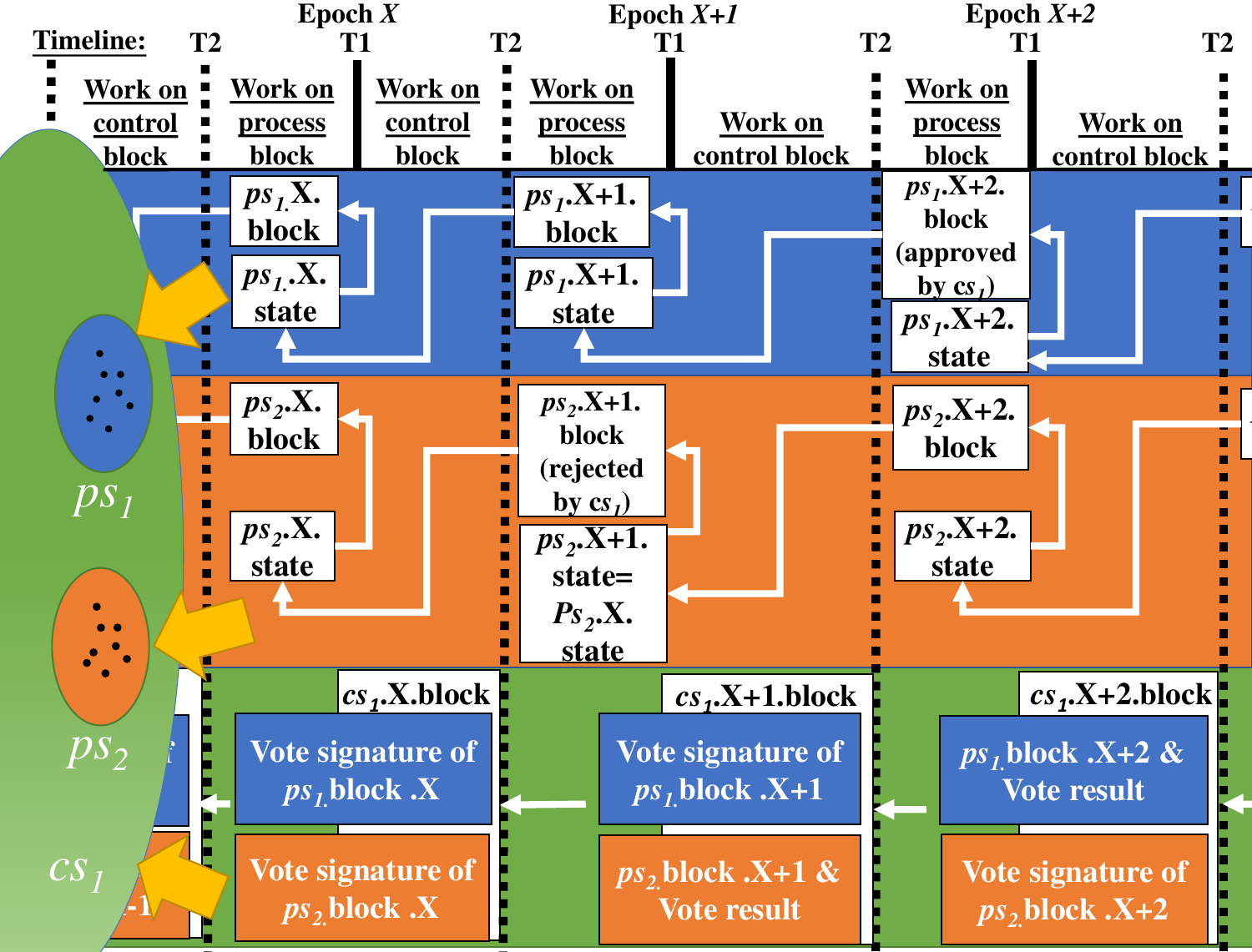}
    \caption{An example of the two-layer consensus. In this example, the control shard $cs_1$ contains two process shards, $ps_1$ and $ps_2$. We denote the state and block of $ps_i$ for Epoch $X$ by ``$ps_i$.X.state'' and ``$ps_i$.X.block'' respectively. A new state of $ps_i$ is derived by applying the transactions in $ps_i$.X.block over $ps_i$.X.state. In the depicted scenario, $ps_1$ and $ps_2$ generated unanimously approved blocks in Epoch $X$, but $ps_2$ failed to agree on the process block of Epoch $X+1$ unanimously. Therefore, $ps_2$.(X+1).block is attached to $cs_1$.(X+1).block, which denotes the block of $cs_1$ for epoch $X+1$. Nodes in $cs_1$ vote on $ps_2.X+1.block$. In this case, $cs_1$ rejects this block, so $ps_2.X+1.state$ remains the same as $ps_2.X.state$. Note that $T_2$ in Epochs $X+1$ and $X+2$ are longer than that in Epoch $X$ because some process blocks did not pass unanimous voting. $T_2$ in Epoch $X+1$ and $X+2$ are longer than that in Epoch $X$ because there are process blocks not passing the unanimous voting. }
    \label{fig:overview}
\end{figure}
Fig.~\ref{fig:overview} provides a visual illustration of the process and control shards in Reticulum. The pseudo-code for the second phase is given in Alg.~\ref{second-phase}.

\begin{algorithm}
\caption{Second phase of Two-layer consensus}\label{second-phase}
\begin{algorithmic}[1]
\Procedure{SecondPhase}{$node$}\Comment{$node$ run this procedure}
\If{$node \in  failed process shard$}
\State send last state and the process block to nodes in the same control shard
%\Else
%\State pass
\EndIf
\If{$node$ is the leader node in control shard}
\For{$B \in$  process shards}
\State VerifyProcessBlock($B$)
\EndFor
\State create and broadcast the control block accordingly
\EndIf
\State vote for the control block
\If{control block is accepted by majority within $T_2$} 
\For {node $i \in$ control shard}
\State confiscated node $i$'s PoS and expel node $i$ if node $i$ is considered the adversarial node.
\EndFor
\State initiate new epoch
\Else
\State repeat second phase with same $T_2$
\EndIf
\If{process block of epoch $X+1$ is rejected in control block}
\State state for process shard in epoch $X+1$ is set to be identical to state in epoch $X$
\EndIf
\EndProcedure
\Function{VerifyProcessBlock}{B}
\If{$B$ passed unanimous voting within $T_1$}
\State attach signatures of voters to $B$ to the control block
\State attach boolean parameter indicating $B$ passed unanimous voting to the control block
\Else
\State attach boolean parameter indicating $B$ failed unanimous voting to the control block
\State attach decision to accept or reject $B$ to the control block
\EndIf
\EndFunction
\State \textbf{Calculate $T_2$:}
\State $\lambda$ - constant value
\State $N_s$ is the count of succeeded process shards in control shard
\State $T_2 \gets \lambda \cdot (\lfloor N_c/N_p \rfloor - N_s+1)$
\end{algorithmic}
\end{algorithm}

\begin{corollary} \label{l1}
    Even if the adversaries have a node in each process shard, the system can resist attacks on liveness while maintaining safety.
\end{corollary}

\begin{proof}
A node is expelled from the system and its PoS be confiscated if it proposed a process block containing faulty information (when it is a leader node), or it voted for a process block containing faulty information or it did not vote at least $\tau-1$ times in every $\tau$ rounds of voting. Consider two types of adversarial behaviors:

 1) \textbf{Adversaries aiming to retain PoS:} To avoid losing their PoS, adversaries can only remain silent once in every $\tau$ round of first-phase voting. By setting $\tau$ appropriately, the number of process blocks progressing to the second phase can be restricted. This restriction inhibits the ability of adversaries to influence the consensus outcome and preserves the integrity of the system.

 2) \textbf{Adversaries wishing to halt system progress:} Adversaries may attempt to halt the system progress by leaving the system permanently, even if their PoS is confiscated by the system. The loss of a relatively small $\beta$ number of PoS (one adversarial node per process shard), is outweighed by the potential gain of bringing the system to a complete stop. Although such attacks could be detrimental, the system is designed to expel these nodes at the second epoch since the attack and recover liveness by seeking unanimous consensus among the remaining nodes in the shard.

It can be inferred that the expelled nodes, which fail to participate in the system at every $\tau$, are not classified as one of the $\left\lfloor \frac{2(N-1)}{3} \right\rfloor + 1$ constant uncorrupted participants nodes. By referencing Def.~\ref{Theorem} (the liveness guarantee theorem),  given that the system maintains $\tau$ liveness guarantee for the honest nodes, there are sufficient honest nodes to ensure the constant safety and liveness guarantee of the control shards after deleting the expelled nodes. Consequently, the process shards can safely seek unanimous consensus among the remaining nodes.

Hence, the two-layer consensus protocol resists attacks on liveness and will always recover from the attacks while upholding safety. 
\end{proof}

\begin{theorem}
    The two-layer consensus protocol guarantees safety, ensuring that only blocks approved by a unanimous vote in the first phase or by a majority vote in the second phase are accepted, and rejected blocks do not affect the state evolution in subsequent epochs. 
\end{theorem}
\begin{proof}
    Please see the full proof in Appx.~\ref{a2}.
\end{proof}

\subsection{Cross-shard transaction}\label{cross}
Cross-shard transactions involve the exchange of records between different process shards. Because nodes only know transactions within their own shards, the cross-shard transaction requires verifying the authenticity and timeliness of records from other process shards.

A shard in Gearbox may experience deadlock, so Gearbox requires liveness detection when doing cross-shard transactions. Reticulum does not have deadlock shards. The transactions in the process block is finalized and become usable when the corresponding control block is accepted. If consider the control shard as a shard in other sharding approaches, it allows direct plug-in for various cross-shard transaction protocols used in the classical sharded blockchains. 
Several approaches exist for processing cross-shard transactions. Omniledger~\cite{kokoris2018omniledger} utilizes a client-driven method that collects and transmits availability certificates. RapidChain~\cite{zamani2018rapidchain} divides cross-shard transactions into single-input single-output transactions and commits them in a specific order. Additionally, Monoxide~\cite{wang2019monoxide} employs relay transactions for cross-shard processing. These methods can be seamlessly applied in Reticulum.   

We show one possible design for implementing cross-shard transactions, which closely replicates the design implemented by Rapidchain. In particular, each cross-shard transaction includes unique participant identifiers, such as one or more wallet addresses for sending money (or any recorded message) and one or more addresses for receiving the funds. Each cross-shard transaction can be decomposed into a set of transactions involving a single sender address and a single recipient address in different shards. We refer to these individual transactions as $Tx_{cross}$. The $Tx_{cross}$ is emitted by $ps_{send}$ and send to $ps_{receive}$ which representing the shard containing the sender address and the shard containing the recipient address respectively. Alg.~\ref{a4} in the appendix provides the pseudocode for this process. Note that we do not claim novelty for the cross-shard transaction designs.  \\

\begin{lemma}
The cross-shard transaction mechanism in Reticulum ensures the security and integrity of transactions between different process shards, providing transaction validity, consensus participation, proof of transaction, and a fixed shard membership approach.
\end{lemma}
\begin{proof}
    Please see the full proof in appendix~\ref{a3}.
\end{proof}

\subsection{Shard size and time-bound}
\label{5.1.1}
This section calculates the process shard size $N_p$, the control shard size $N_c$ and the settings for time-bound.\\

\noindent \textbf{Set the value of $N_p$.} 
Let $N$ be the number of nodes and $N_{p}$ be the number of nodes in a process shard. The system has $\lfloor N/N_p\rfloor$ process shards and each sized at least $N_p$. 
Because nodes are randomly assigned to the shards, to fulfill a given $P_{fp}$, where $P_{fp}$ represents the probability for a \textit{particular} process shard to be compromised, it yields the equation:
\begin{equation}
N_p=\log_{P_a}(P_{fp})
\end{equation} where the maximum ratio of adversarial population in the system is $P_{a}$.\\

\noindent \textbf{Set the value of $N_c$.} Let $N_{c}$ denote the number of nodes in a control shard. The control shard's safety threshold is $S=L<50\%$. A control shard can output illegal blocks if adversarial nodes control more than $\lfloor\frac12N_c\rfloor$ nodes. The failure rate of a control shard is:
\begin{equation}  P_{fc}=\sum_{i=\lfloor\frac12N_c\rfloor+1}^{N_c}\binom{N_c}{i}{(P_a)}^i{(1-P_a)}^{N_c-i}\end{equation}

With a given $P_{fc}$ and $P_a$, $N_{c}$ can be determined accordingly.
Tbl.~\ref{tb1} exhibits the size of the process and control shards with different $P_a$, $P_{fp}$ and $P_{fc}$.\\

\begin{table}[h!]
\centering
\begin{threeparttable} \small 
\caption{Process shard size and control shard size}
\label{tb1}
\begin{tabular}{cccccc} \hline
\diagbox{\footnotesize $P_{fp},P_{fc}$}{$N_{p}\&N_{c}$}{$P_{a}$}&15\%&20\%&25\%&30\%&33\%\\ \hline
${10^{-5}}$&7\&27&8\&41&9\&63&10\&105&11\&149\\ \hline
${10^{-6}}$&8\&35&9\&51&10\&79&12\&131&13\&185\\ \hline
${10^{-7}}$&9\&41&11\&61&12\&95&14\&155&15\&221\\ \hline
\end{tabular}
\end{threeparttable}
\end{table}

\noindent 
\textbf{Security threshold.} The failure rate represents the probability of any shard (regardless if it is a control or a process shard) outputting illegal blocks and is denoted as $P_f$.\begin{equation}
    P_f<\left ((P_{f})_{threshold}=2^{-\sigma} \right )
\end{equation} with predefined security parameter $\sigma$.\\

\noindent \textbf{Determining $P_{fp}$ and $P_{fc}$.} 
Given $(P_f)_{threshold}$ and $N$, we calculate the approximate $P_{f_{approx}}$ bounded by:
\begin{equation}
    P_f<\left (P_{f_{approx}}=P_{fp}\times \beta +P_{fc}\times \Lambda \right )<(P_f)_{threshold}
\end{equation} where $\beta =\lfloor N/N_p\rfloor$ and $\Lambda=\lfloor N/N_c\rfloor$. We want to minimize $N_p$ and $N_c$ to allow as many shards as possible. Tbl.~\ref{tb3} shows several $N_p$ and $N_c$ with the given $P_{f_{approx}}$.\\
\begin{table}[h!]
\centering
\begin{threeparttable}
\caption{Process shard size and Control shard size, $P_a=33\%$, calculated using Alg.~\ref{approx}}
\label{tb3} 
\begin{tabular}{cccccc}\hline
\diagbox{\footnotesize $P_{f_{approx}}$}{$N_{p}\&N_c$}{$N$}&500&1000&5000&10000&20000\\ \hline
${10^{-5}}$&15\&221 &15\&221&17\&257&17\&257&19\&293\\ \hline
${10^{-6}}$&17\&257 &17\&257&19\&293&19\&293&21\&329\\ \hline
${10^{-7}}$&19\&293 &19\&293&21\&329&21\&329&23\&367\\ \hline
\end{tabular}
\end{threeparttable}
\end{table}

\begin{algorithm}[h!]
\caption{Determining $P_{fp}$ and $P_{fc}$ with given $P_a$ and $P_{f_{approx}}$}\label{approx}
\begin{algorithmic}[1]\small

\Function{Get$N_p$}{$P_{fp}$}
\State \textbf{return} $\lceil \log_{P_a}(P_{fp})\rceil$
\EndFunction

\Function{Get$N_c$}{$P_{fc}$}
\State \textbf{return} the smallest $N_c$ that failure rate less than $P_{fc}$
\EndFunction

\Function{Get$N_c$ and $N_p$}{$P_{f_{approx}}$}
\State $\arg \max_{i} f(i)=\{i\times\lfloor N/GETN_p(i)\rfloor+i\times\lfloor N/GETN_c(i)\rfloor|f(i)\leq P_{f_{approx}}\}$
\State \textbf{return} $GETN_p(i)$ and $GETN_c(i)$ 
\EndFunction

\end{algorithmic}
\end{algorithm}

\noindent \textbf{Parameters for time-bounds.} Setting up the time-bound for synchronous protocols is an open question for protocol designers. A large time-bound delivers redundancy, while a short time-bound may cause some nodes to be unable to catch up on the communication of the network. Previous work, e.g., Rapidchain runs a pre-scheduled consensus among committee members about every week to agree on a new time-bound $\Delta$. Thus, in theory, Reticulum can use the same design to agree on new $T_1$ and $\lambda$.

\section{Analysis}
\label{experiment}
This section shows an in-depth analysis for Reticulum.
\subsection{Security Breaches}\label{breaches}
This section analyzes a possible equivocation of consensus associated with the design that leverages the liveness ($L$) and safety ($S$) of the shards. 
Indeed, the adversary cannot force the acceptance of a block containing incorrect information if the system only accepts a block when at least $(S\times M)+1$ nodes (synchronous version) or at least $(\max(2L, S)\times M)+1$ nodes (partially synchronous version)  in a shard sized $M$ voted to accept it. This is because the system only allows at most $S$ ratio of adversarial population in a shard (the probability for a shard to have more than $S$ ratio of adversarial population is negligible.)

However, we prove in Lem.~\ref{llemma1} that synchronous shards with $S \geq \frac{1}{2}$ or partially synchronous shards with $S \geq \frac{1}{3}$ are unable to achieve consensus without the possibility of equivocation, which can lead to security breaches. This issue also affects Reticulum because the process shards of Reticulum maintain $L=0\%$ and $S<100\%$.

\begin{lemma}
Synchronous shards with $S \geq \frac{1}{2}$ or partially synchronous shards with $S \geq \frac{1}{3}$ are unable to achieve consensus without the possibility of equivocation.
\label{llemma1}
\end{lemma}

\begin{proof}
The maximum number of adversarial nodes allowed in a shard is $f = S \cdot M$, where $S < 1 - L$ in the synchronous model and $S < 1 - 2L$ in the partially synchronous model, and $S \geq L$ at all times.

When a block containing correct information is being voted upon, the security assumption only guarantees that honest nodes receive the same set of $M - f$ honest votes from each other. Because $M - f < (S \times M) + 1$ in the synchronous model and $M - f < (\max(2L, S) \times M) + 1$ in the partially synchronous model, it would require adversarial nodes to vote to accept the block. It is possible for adversarial nodes to send inconsistent votes to different nodes. This implies that not all honest nodes will consider the block as having received enough votes for approval, which contradicts the termination and agreement properties of the consensus protocol, which require all honest nodes to commit to the same value before the protocol terminates.

Consider the phases of voting in the consensus protocol (e.g., Pre-prepare, Prepare, Commit, etc.). For illustration, assume all nodes are in the final phase before accepting a block. In a synchronous example, an adversary convinces an honest node, Alice, that block $A$ has received $(S \times M) + 1$ votes in favor of acceptance, so Alice believes the block is accepted. However, all other honest nodes only received $S \times M$ votes in favor of acceptance, so they believe the block is rejected. After the time-bound expires, the next leader node proposes a new block $B$ to replace $A$. Since Alice believes $A$ has been accepted, she will not vote to accept $B$. However, it is still possible for $B$ to obtain $(S \times M) + 1$ votes in favor of acceptance afterward, as only Alice believed $A$ was accepted.

Ideally, Alice could attach the vote results she previously received when voting to reject $B$ because $A$ was accepted first, or inform others when she became aware that $A$ was accepted. However, if Alice is an adversary, she could present the vote results of $A$ after block $B$ has already received $(S \times M) + 1$ votes for acceptance, or even after new blocks have been built on top of $B$.

This demonstrates that although adversarial nodes cannot enforce the acceptance of adversarial blocks, they can still cause equivocation (both $A$ and $B$ received $(S \times M) + 1$ votes). The reason for this is that $M - f < (S \times M) + 1$ in the synchronous model and $M - f < (\max(2L, S) \times M) + 1$ in the partially synchronous model, indicating a lack of sufficient votes commonly known to all honest nodes.

Because nodes outside the shard did not witness the voting within the shard, they cannot determine whether $A$ or $B$ received enough votes first, preventing the system from reaching an unequivocal consensus.

One might argue that HotStuff~\cite{yin2019hotstuff} uses the leader node to transmit the votes with threshold signatures, thereby avoiding vote discrepancies. However, to achieve unequivocal consensus, it would require that at least $(S \times M) + 1$ or $(\max(2L, S) \times M) + 1$ nodes, depending on the communication model, follow the voting rules between phases, which is unachievable.
\end{proof}

Thus, for these shards to function safely without external censorship, each honest node would need to initiate an instance of Byzantine Reliable Broadcast (BRB) to share the votes it has received. However, the partially synchronous BRB only functions with $f\leq \lfloor (M-1)/3 \rfloor$~\cite{abraham2021good}, which contradicts the setting of $f$ in the shard. The synchronous BRB allows for an adversarial majority, but it requires $(\frac{M}{M-f}-1)\Delta$ per vote, where $\Delta$ is the known communication time-bound. Therefore, only synchronous shards with $S\geq \frac{1}{2}$ may function independently, but they take a long time to finalize the voting process.

Gearbox uses a design in which each shard regularly posts a heartbeat (a collection of vote signatures) to a ``control chain'' that includes all nodes in the system. This design serves solely the purpose of detecting liveness. However, a new epoch of a shard does not wait for such a heartbeat to appear in the block of the ``control chain.'' This oversight disregards the possibility of equivocation views among honest nodes. 

Since there is no consensus on the set of vote signatures (heartbeats), the leader node may mistakenly consider a shard as dead while other honest nodes believe otherwise. These circumstances may lead to multiple iterations of the consensus-reaching process in the ``control chain'' until the honest nodes are aligned. In the worst situation, the system's security will be compromised as suggested by Lem.~\ref{fupdate}.

\begin{lemma}
The equivocation problem is contagious; finalizing a decision with the possibility of equivocation is not safe. \label{fupdate}
\end{lemma}

\begin{proof}
The validity property of any binary consensus protocol (e.g., PBFT~\cite{castro1999practical}) requires all honest nodes to hold the same input value before consensus on the input value can be reached. As demonstrated in Lem.~\ref{llemma1}, honest nodes may encounter equivocation, thus not holding the same input value.

Some equivocation can be resolved: for example, assume that an honest node, Alice, knows that an acceptance decision for a block $A$ of a shard has been reached, while another honest node, Bob, knows that an acceptance decision for a different block $B$ of the shard has been reached. In this scenario, if the leader node of the control chain proposes to accept either $A$ or $B$ with the votes attached, both Alice and Bob can safely accept the block suggested by the leader node, as more than a certain ratio $S$ of nodes have approved the decision.

However, some equivocation cannot be resolved. For instance, if the leader node claims that no decision has been received in the shard, neither Alice nor Bob will accept the control chain block because they know that decisions were made. Nevertheless, if all other honest nodes did not receive a decision from the shard and voted for the acceptance of the control chain block, and if two adversary nodes also voted for the acceptance of the control chain block, the control chain block can still be approved even though Alice and Bob did not accept it. This results in an equivocation problem because there is no guarantee that the delivery of the two adversarial votes will arrive on time to all honest nodes.

Therefore, the system cannot safely decide on a control chain block that contains equivocated decisions.
\end{proof}

The only way to avoid the unsolvable equivocation problem discussed in the proof of Lem.~\ref{fupdate} is to approve whatever the leader node of the control chain proposes, as it is impossible to deterministically decide if a shard has lost liveness. However, this approach allows an adversarial leader node to arbitrarily rebuild a shard that, in reality, has not lost liveness, thereby introducing additional safety vulnerabilities.

\begin{theorem}
Reticulum is secure when operating with $L=0\%$ and $S<100\%$ process shards and $L,S<50\%$ control shards.
\end{theorem}

\begin{proof}
Reticulum does not suffer from the security breaches indicated in Lem.~\ref{llemma1} because the process shards do not function independently. The blocks of the process shards in Reticulum are confirmed in the corresponding control shard before the next epoch starts. The votes of the process blocks are broadcasted using synchronous BRB with $f\leq \lfloor (N_c-1)/2\rfloor$ to all nodes in the corresponding control shard. Therefore, the honest nodes in the control shard are natively aligned and it is not possible to have consensus equivocation. We show in Sec.~\ref{BRB} that such a protocol only requires $4\Delta$ in overall and at the communication level of $O({N_c}^2)$. 
\end{proof}

Reticulum maintains the control shards instead of a ``control chain'' of Gearbox. This is to ensure that the vote broadcasting only involves $N_c$ nodes, instead of $N$ nodes.

\subsection{Analysis of Overhead in Byzantine Reliable Broadcast Protocol Utilization}\label{BRB}
In existing sharding protocols, they adopt BFT protocols like PBFT~\cite{castro1999practical} to build replicae, which incurs $O(M^2)$ communication complexity, as it requires nodes to forward the proposal (block) to avoid equivocation. There is no need to forward votes and is safe to make a decision when a block receives $f+1$ identical votes (assume synchronous communication) as it is impossible for a different verdict to receive more than $f$ votes. When a consensus decision is made for a proposal, it only guarantees that the votes exceeded $f$, but not a consensus for the exact set of votes. 

The first phase of Reticulum requires a determined set of votes to determine adversarial nodes, therefore, each node $i$ utilizes a Byzantine reliable broadcast protocol called $(\Delta+\delta)$-BRB~\cite{abraham2021good} to broadcast its vote $V_i$, which requires four steps: 

\begin{enumerate}
    \item Step 1: The broadcaster (node $i$) sends the proposal ($V_i$) with $O(N_c)$ complexity.
    \item Step 2: Everyone votes for and re-transmit $V_i$ with $O({N_c}^2)$ complexity.
    \item Step 3: Everyone commits and locks with $O({N_c}^2)$ complexity.
    \item Step 4: A Byzantine agreement is conducted, which incurs an overhead of $O({N_c^2})$.
\end{enumerate}

Therefore, $(\Delta+\delta)$-BRB demands $4\Delta$ to complete and the communication complexity for determining a $V_i$ in consensus is at the level of $O({N_c}^2)$. Thereby, the overall complexity for the BFT process in Reticulum reaches $O({N_c}^3)$, which is a big overhead.

To optimize this, we design that step 2 for all instances of $(\Delta+\delta)$-BRB are combined together. It is safe to do so as step~1 must end at a fixed time-bound $\Delta$. In this combined design, a vote in step~2 is a collection of votes for the proposals in step~1. When step~2 are combined, steps~3 and 4 are automatically combined.

Then Raticulum's two-layer consensus incurs the same complexity level: $O({N_c}^2)$ for send and forwards the proposal and $O({N_c}^2)$ for send and forward all the votes for the proposal. Therefore, Raticulum's communication overhead is still at the level of $O({N_c}^2)$, just with a bigger constant than PBFT or other BFT protocols.
\subsection{Attacks on liveness with different $\tau$}
This section analyzes the trade-off between the requirement for the adversarial nodes and the system performance. When the adversary has more than $\tau$ nodes in a process shard, the adversary can permanently stop the process shard without being expelled from the system. This is achieved by each adversarial node taking terms to not participate in voting in the first phase of the two-layer consensus or the adversarial leaders propose incorrect blocks. Therefore,  $\tau\geq N_p$ should be set to guarantee the eventual liveness of the process shard. The rate for successfully accepting a process block within the process shard is
\begin{equation}
    R_p=\frac{\tau-a}{\tau}\label {etau}
\end{equation}
where $\tau\geq N_p$ and $a$ is the adversary nodes in the process shard. We allow $N_p-1$ adversarial nodes in this shard in the worst case. Thus, considering the worst case, $R_p=\frac{\tau-N_p+1}{\tau}$.

Fig.~\ref{fig:tau} shows different $\tau$ corresponding to different $N_p$ and $R_p$. We show some numerical results: considering the worst case $R_p$, to maintain a success rate $R_p=40\%$, $R_p=70$ or $R_p=90\%$, an adversarial node in a process shard sized $N_p=15$ may go offline for one epoch in each 24, 47 or 140 epochs respectively. This is a reasonable relaxation compared to the constant liveness requirement.

\begin{figure}[h!]
    \centering
    \includegraphics[width=0.5\textwidth]{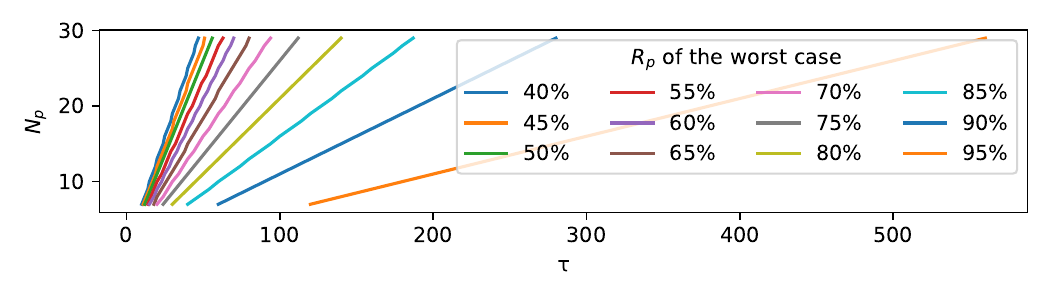}
    \caption{$\tau$ corresponding to $Rp$ of the worst case}
    \label{fig:tau}
\end{figure}

If the adversary wishes to attack liveness and be expelled, it can stop every process shard and the first phase of the overall system for around $2a$ epochs. Afterward, each process shard maintains $R_p=1$. 

\noindent\textbf{Discussion over $\tau$ liveness.}
The general security assumption for Rapidchain assumes that at any given moment, at least 2/3 of the computational resources belong to uncorrupted participants. It infers the constant liveness for the same set of honest participants as shown in Thm.~\ref{tt1}. Without the constant liveness guarantee, the verdict reached is not deterministic, for the similar reason shown in the proof of Lem.~\ref{llemma1}. 

If we wish to apply $\tau$ liveness also to the honest nodes, the verdict reached must go through one round of Byzantine reliable broadcast tolerating a majority adversarial population before the next round of voting may start. which incurs significant costs. Therefore, the same as all the existing approaches, we just assume the constant liveness for the honest nodes.

Admittedly that allowing a larger $\tau$ (inferring a stronger $R_p$) would bring the stricter liveness requirement for the adversarial (unstable) nodes. One practical relaxation, in reality, is to design that the nodes are only being given a long suspension when identified as adversarial nodes. They may participate in voting after the prison term.

\subsection{Analytical performance comparison with Rapidchain }\label{theoratical_rapidchain}
This section provides the performance comparison between Reticulum and Rapidchain. Because there is no source code for Rapidchain available for experiments, this section provides an analytical comparison.

The two-layer consensus of Reticulum will first reject the process block which contains wrong information in the process shard and then again in the control shard. The nodes which proposed these blocks will be marked as adversarial nodes and be expelled. Therefore, in general terms, we would expect that the process blocks are correct, and it is only a matter of being accepted at the process shard or at the control shard. Therefore we may approximate the throughput (the number of transactions processed per second) and the storage requirement (the storage per transaction) as follows, please find some notations for this section in Tbl.~\ref{tab:experimentnotations}:

\begin{table}
\caption{Some Notation for the calculation}
\label{tab:experimentnotations}\small
\begin{tabular}{cl}
Notation &Description                                                                                                                                                                                                                                                                                     \\ 
\hline
$E_{shard}$              & \begin{tabular}[l]{@{}l@{}}$E_{shard}=\lfloor N_c/N_p\rfloor$ denotes the number of process\\shards within a control shard.\end{tabular}                                                                                                                                                                               \\\hline
$E_{time}$               & \begin{tabular}[l]{@{}l@{}}$E_{time}$ denotes the length for one blockchain epoch.\\ $E_{time}=T_1+T_2$ for Reticulum.\end{tabular}                                                                                                                                                                                            \\\hline
$B_s$                    & \begin{tabular}[l]{@{}l@{}}  $B_s=2MB$ denotes the block size of the process\\block. It contains 4096 transactions.\end{tabular}                                                                                                                                                                  \\\hline
$E_{tx}$                 & \begin{tabular}[l]{@{}l@{}}$E_{tx}=E_{shard} \times 4096$ denotes the number of overall\\ transactions per epoch logged to the blockchain \\within a control shard in Reticulum or a shard in\\ Rapidchain.\end{tabular}                                                                                                        \\\hline
$N_i$                    & \begin{tabular}[l]{@{}l@{}}$N_i$ indicates the number of nodes that store the \\$i$-th process block in the current epoch.\\$N_i$ is determined by whether the process shard $i$ \\reached a consensus in the first phase. If the \\process shard unanimously agreed on its process \\block, $N_i=N_p$; otherwise, $N_i=N_c$.\end{tabular} \\
\hline\end{tabular}
\end{table}

\noindent \textbf{Throughput.}{
We calculate \begin{equation}\label{eq7}
    Throughput=\frac{E_{tx}}{E_{time}}
\end{equation} $R_p$ affects the length of $T_2$, so it affects the throughput.
}

\noindent \textbf{Storage.}{ $S_{tx}$ measures the overall storage incurred in the nodes when the blockchain within a control shard logged one transaction. 
We calculate 
\begin{equation} \label{eq8}
    S_{tx}=\sum_{i=1}^{E_{shard}}\frac{B_{s}\times N_{i}+(N_{i}-N_p)\times State_i}{E_{tx}}
\end{equation} where $B_{s}\times N_{i}+(N_{i}-N_p)\times State_i$ calculates the combination of the storage incurred for all the nodes that stored process block $i$. $State_i$ is the size of the current state of process shard $i$. $R_p$ determines $N_{i}$, so it also affects $S_{tx}$.
 }

To illustrate the relationship between $R_p$ and the performance and compare it with Rapidchain, we perform one mathematical simulation. In this simulation, we set $(P_{f})_{threshold}=10^{-7}$, $\lambda=50$, $P_a=33\%$, $N=5000$, $N_c=329$ and $N_p=21$. Therefore, Reticulum has 15 control shards and each control shard has $E_{shard}=15$ process shards inside. Each shard in Rapidchain maintains $L, S<50\%$ so it sized $N_c$ nodes, the same as the control shard of Reticulum. Therefore, Rapidchain has 15 shards. Because Rapidchain only uses one layer consensus, its performance is dependent on the size of its blocks and the number of shards. The two protocols have different settings. In order to relate the two works, we set the same upload bandwidth for broadcasting blocks. We also set the same $E_{tx}$ for both approaches, so a block of a shard in Rapidchain sized $B_s\times N_p$, containing $E_{tx}$ transactions. We make $E_{time}$ different for the two approaches. 

In Reticulum, the upload bandwidth requirement for the nodes sending the process block and the latest states to all nodes of the same control shard is denoted as $UB$. This may happen in the second phase of the two-layer consensus. 

1) \textit{In the worst scenario} when all process shards within the control shard fail to obtain the unanimous approval within $T_1$, then $T_2 =(E_{shard}+1)\times \lambda=800 s$ and $UB = \frac{B_s \times N_c +(N_c-N_p)\times state }{T_2-\Delta}$. We assume that each $state$ size $256 KB$. It is intended not to define the structure of the state in this paper, as that is application-oriented. Here we assume, each state contains $10922$ wallet addresses (160 bits each) and their balances within the process shard (32 bit each). 

2) \textit{In the best case scenario}, all process blocks obtained the unanimous approval in the first phase, then $T_2=50 s$ and there is no need to post any process block to the control shard. Assuming $\Delta=10 s$, the upload bandwidth requirement to send the process block to the control shard in the worst case is, therefore, $UB\approx 952.708 KB/s$. With this upload bandwidth, it is reasonable to set $T_1=\frac{B_s\times N_p}{UB}+4\Delta \approx 86 s$.  Then in the worst case, one epoch for Reticulum lasts $E_{time}=T_1+T_2=886s$. In the best case, it lasts $136s$.

To use the same upload bandwidth requirement for the leader node of Rapidchain, $ UB=\frac{B_s\times N_c}{E_{time}+\Delta}$. $E_{time}=\frac{B_s\times N_c}{UB}-\Delta \approx 698s$, which is similar to $T_2$ of the worst case. In this setting, Rapidchain has a constant storage per transaction $S_{tx}=\frac{B_s\times N_c\times N_c}{E_{tx}}\approx 3608 KB$, it does not need to sync the states and has a constant transaction per second of $\frac{E_{tx}}{E_{time}}\approx 88.022 tx/s$ for a shard and $1320.330 tx/s$ in the overall system (15 shards).  Note that $S_{tx}$ refers to the overall storage incurred in the network instead of the storage incurred for an individual node. 

\begin{figure}[h!]
    \centering
    \includegraphics[width=0.5\textwidth]{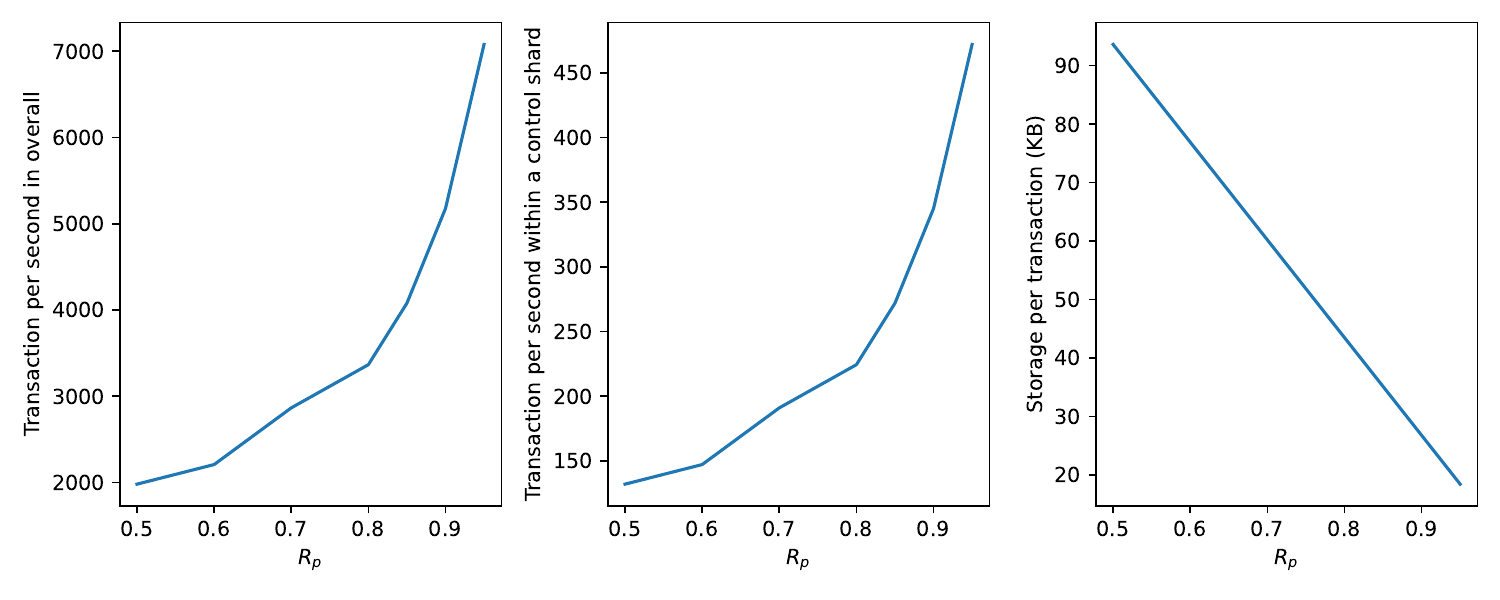}
    \caption{The theoretical throughput and storage considering different $R_p$, with $UB=952.708 KB/s$}
    \label{fig:theoratical_rapidchain}
\end{figure}

Fig.~\ref{fig:theoratical_rapidchain} shows the simulation for the theoretical performance of Reticulum with different $R_p$. $E_{time}$ is calculated according to the given $UB$. It significantly outperformed Radpichain as, in theory, Rapidchain can be seen as the design when Reticulum deteriorates to the one-layer design that all process shards constantly failed. With $P_a=33\%$, $N_p=21$ and if set $\tau=40$, we get $R_p=50\%$ (assuming the worst case). In this case, the transaction per second for the overall system is $1980.5 tx/s$, approaching the double performance of Rapidchain and only incurs $93.62 KB$ storage per transaction overall in the network. When setting $\tau=411$, we get $R_p=95\%$ (assuming the worst case). In this case, the transaction per second for the overall system is $7077.645 tx/s$ and only incurs $18.4099 KB$ storage per transaction (as more transactions are only stored in the process shard) in overall in the network. Fig.~\ref{fig:theoratical_rapidchain} also confirms the general assumption of the larger $R_p$ brings the stronger performance.   

Note that we do not consider cross-shard transactions. As mentioned in Sec.~\ref{cross}, we may plug in different cross-shard transaction protocols and may use exactly the same design as Rapidchain. There are no different properties in this field between the two protocols.
\begin{figure*}[h!]
    \centering
    \includegraphics[width=\textwidth]{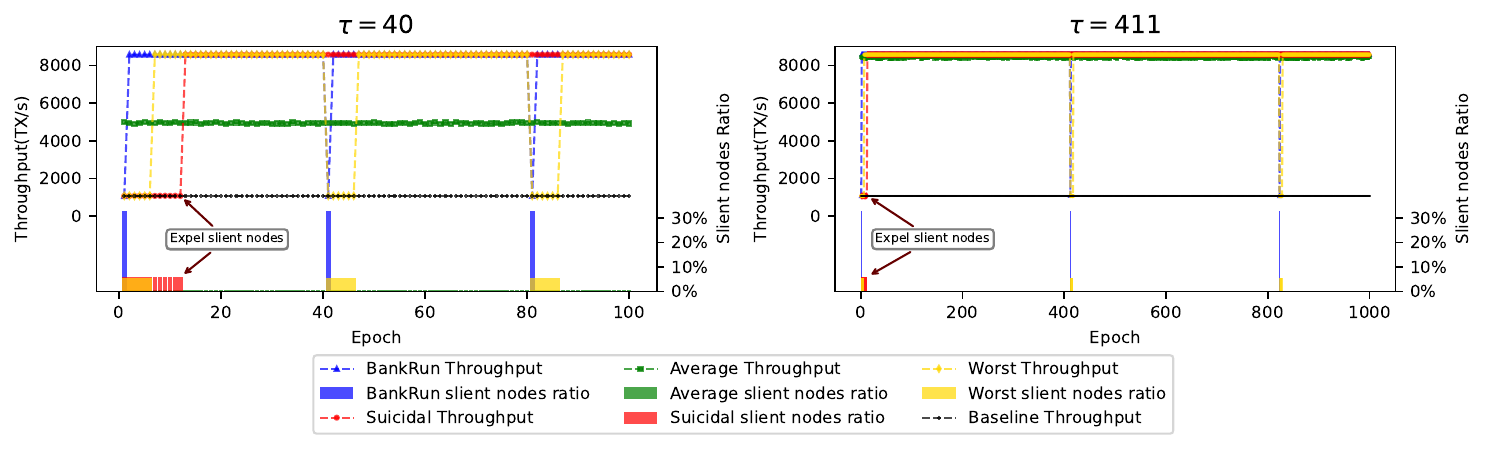}
    \caption{The experiment result of a system of Reticulum with $(P_{f})_{threshold}=10^{-7}$, $N=5000$, $N_c=329$, $N_p=21$, $P_a=33\%$, $T_1=86s$, $\lambda=800$, $B_s=2MB$ (a block has 4096 transactions) and different $\tau$ and the result of Baseline with $(P_{f})_{threshold}=10^{-7}$, $N=5000$ and $E_{time}=698s$ and each shard sized $329$. Silent nodes stand for the nodes (globally not just within a shard) that did not vote for the process shard at the given epoch. The experiment setup matches the simulation setup we did in Sec.~\ref{theoratical_rapidchain}. As can be seen from the picture, Reticulum outperforms Baseline in all cases. When $\tau$ is larger, the attack $worst$ occurs less frequently. We need at least $\tau$+1 epochs to illustrate BankRun attacks. Therefore, the experiment last $100$ epochs for $\tau=40$  and $1000$ epochs for $\tau=411$. }
    \label{fig:RapidchainThroughput}
\end{figure*}

\section{Experiment}
This section provides a combination of experiments and mathematical calculations to evaluate the key characteristics of Reticulum and compare it with a Rapidchain-like protocol and also with Gearbox. The Rapidchain-like protocol (baseline) refers to an ordinary one-layer synchronous sharding protocol. We implement both baseline and Reticulum in Golang. Gearbox was evaluated purely via simulations on mathematical models since the authors provided no source code, and some detail designs not related to sharding are not presented, rendering experiments infeasible.

\noindent \textbf{Experiment setup.} We experiment Reticulum and baseline but simulate Gearbox. It is fair (maybe a little unfair for Reticulum) to make a comparison in this way, as we set the communication delays used in the mathematical simulation the same as we observed in the experiment environment for Reticulum. Our experiment was conducted on fifteen servers each equipped with 32-core AMD EPYC 7R13 processors, providing 128 vCPUs running at 3.30 GHz, 256 GB of memory, and a network bandwidth of 10 Gbps. We added a 200ms delay to each message to simulate the geographic distribution of nodes.

\noindent \textbf {Attack strategies.} We consider four kinds of attacks from the adversary. $BankRun$, where all adversarial nodes do not vote for the process blocks at a single epoch. $BankRun$ can only occur once in every $\tau$ epochs. $Average$, where each adversarial node does not vote once in a random epoch in every $\tau$ epochs. $Worst$, where only one adversarial node refuses to vote at each process shard in every epoch. The adversary can stop a process shard for $i<\tau$ epochs in every $\tau$ epochs where $i$ is the number of adversarial nodes inside this shard. $Suicidal$, is based on the $worst$ strategy but all adversarial nodes vote at most $\tau-2$ epochs in every $\tau$ epoch, and be expelled at the second time when they remain silent in voting. \\

\noindent \textbf{Comparison between Reticulum and Baseline.} Fig.~\ref{fig:RapidchainThroughput} suggests that Reticulum has superior performance.\\

\noindent \textbf{Comparison between Gearbox and Reticulum.} Gearbox assumes that the adversarial population in the system ($P_{a-run}$) is static, unknown but below $P_a$ where $P_{a-run}\leq P_a\leq 33\%$. It continually refines shard sizes and liveness thresholds until an optimal arrangement is attained, where the adversarial population is conquered.

To emulate the outcomes of the Gearbox system, we instantiate a network comprising 5000 nodes, with a portion represented by $P_{a-run}$ being designated as adversarial. These adversarial nodes are distributed randomly within the network. Nodes are subsequently assigned to shards through a random allocation process. We then assess each shard to determine whether the percentage of adversarial nodes exceeds the shard's predefined liveness threshold. In cases where this threshold is surpassed, we dissolve the existing shard and construct a new one with a higher gear setting. This process is iteratively executed until all shards maintain a percentage of adversarial nodes that falls below their respective liveness thresholds. 

Gearbox only discussed the mechanisms for when and how to rebuild a shard and to enlarge its size, which is very high level and lacks detailed designs. Therefore, we may not simulate the performance when rebuilding shards. But we do know how many overlapping shards a node is in after shards switch gears. Gearbox recommend the usage of four gears corresponding to liveness values of $10\%$, $20\%$, $25\%$ and $30\%$ within a shard. To align with our work, we additionally use a gear of $49\%$. The shard size corresponding to these gears are $26$, $39$, $50$, $63$ and $293$ respectively, this is calculated considering $P_a=33\%$. We repeat this entire procedure 1000 times, observing the distribution of shards across different gear configurations, as illustrated in Fig.~\ref{fig:overlapping}. Fig~\ref{fig:overlapping} also shows the number of shards in Gearbox that a node is simultaneously located in (overlap times) with the given adversarial ratio in runtime. When a node is in multiple shards, it has multiple workload.

With a given $P_{a-run}$ ratio of adversarial nodes globally, we simulate the following process: At the beginning, every shard is built at the same size using $L=10\%$. We check if a shard contains more than $L$ of adversarial nodes, if so we rebuild it with a larger gear. This process is repeated until there is no shard that has an adversarial population of more than its gear. We then know the overlapping situation. We calculate the performance accordingly. Fig.~\ref{fig:GearboxThroughputfull} shows, with different $P_{a-run}$, the comparison between Reticulum and Gearbox after every shard of Gearbox has found a stable gear.

\begin{figure}[h!]
\centering
\includegraphics[width=0.5\textwidth]{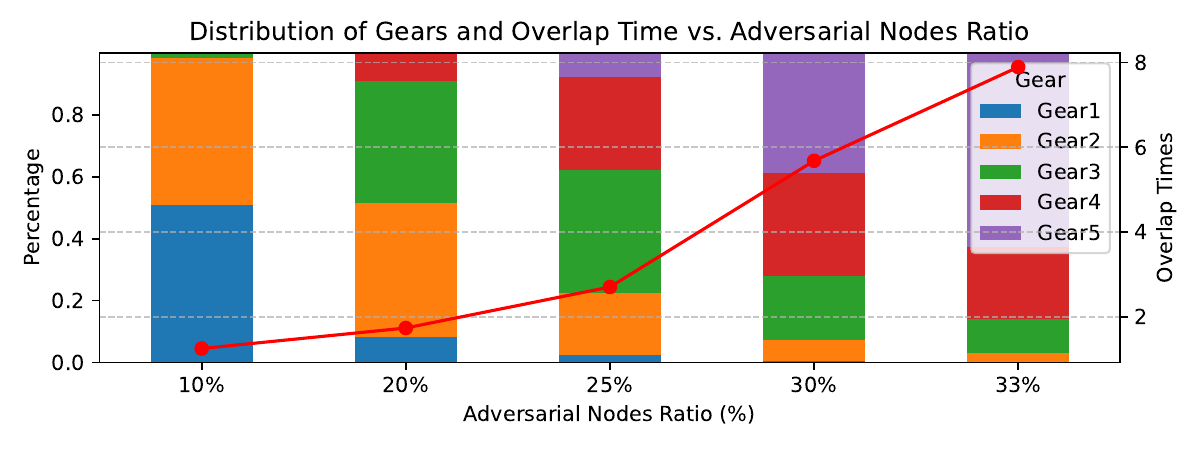}
\caption{The number of overlapping shards and the distribution of gears with different adversary ratio globally}
\label{fig:overlapping}
\end{figure}

As demonstrated in Sec.~\ref{breaches}, the shards within the Gearbox ecosystem face a critical challenge in achieving independent functionality for the majority of the time. Consequently, these shards must patiently await the finalization of their respective blocks within the control chain before commencing the next epoch. Failing to do so would expose them to the risk of consensus equivocation.

\begin{figure*}
    \centering
    \includegraphics[width=\textwidth]{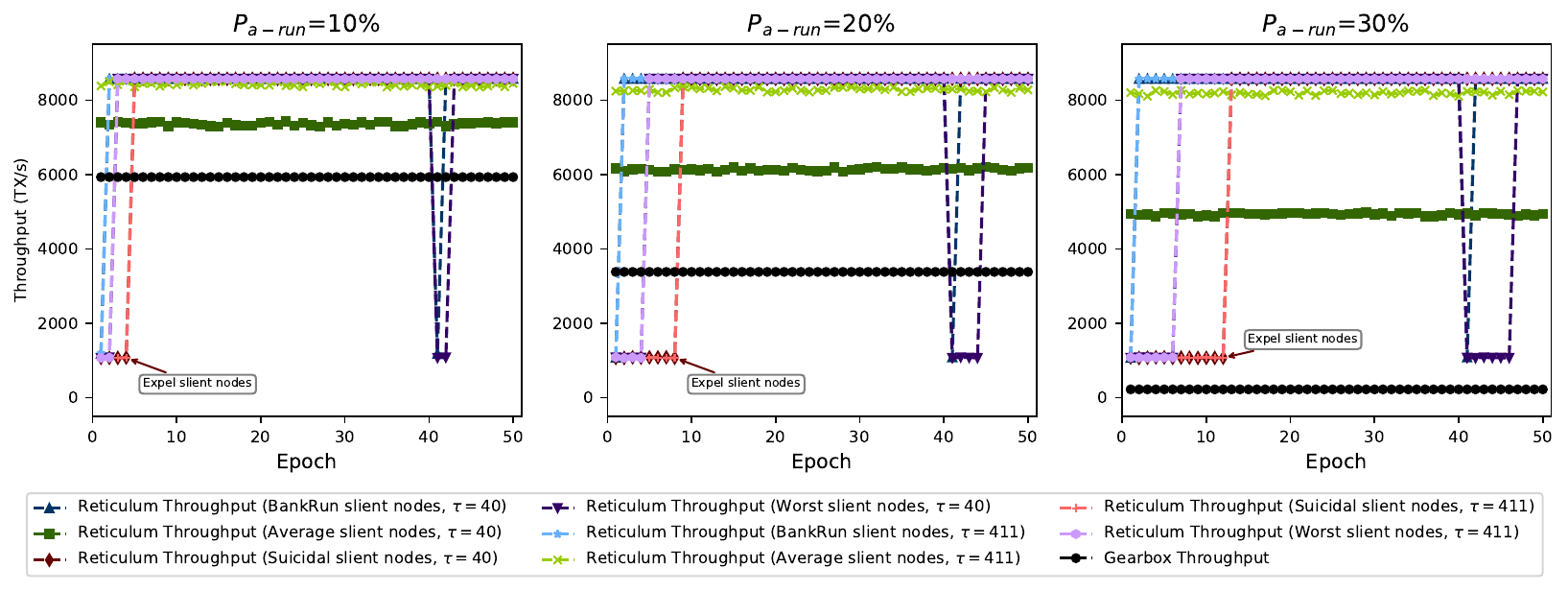}
    \caption{The experiment result of a system of Reticulum with $(P_{f})_{threshold}=10^{-7}$, $N=5000$, $N_c=329$, $N_p=21$, $T_1=86s$, $\lambda=800$, $B_s=2MB$ (a block has 4096 transactions) and $\tau=40$. Because Reticulum has methods to mitigate attacks on liveness and Gearbox does not have such features, all $P_{a-run}$ nodes will actively attack in Gearbox, but will attack according to the attack strategies in Reticulum as we outlined previously. The figure also shows the result of Gearbox with $(P_{f})_{threshold}=10^{-7}$, $N=5000$ and $B_s=2MB$ (a block has 4096 transactions) after each shard has found a stable gear. In the biggest gear $L=49\%$, a shard contains $329$ nodes, which reaches the same size of a control shard of Reticulum. 
    }
    \label{fig:GearboxThroughputfull}
\end{figure*}

One potential remedy involves configuring the block interval of the control chain to be shorter than that of each individual shard. This approach ensures that once a shard has accepted a block, and this decision has reached at least one honest node within the shard, the current leader node of the control chain can be promptly informed of this decision. This, in turn, allows for the incorporation of this decision into the next control block, substantially mitigating the risk of equivocation.

However, reducing the block interval of the control chain comes with a significant drawback: it imposes a substantial communication burden on network nodes. This is due to the fact that the acceptance of a control block necessitates the involvement of the entire system's nodes, which can be a considerable number. To address this issue, we have devised a design in which the entire system outputs only one control block for each block height. This ensures that the length of the blockchain within the shards remains identical to that of the control chain. Consequently, each shard must await the acceptance of the next control block before initiating a new blockchain epoch.

It is worth noting that in this simulation, we make the simplifying assumption of zero communication cost, implying that once the last shard reaches consensus, the control block is instantaneously generated and accepted by all nodes—a scenario that may be unrealistic in practice. Also, we should acknowledge that we do not account for the possibility of adversarial behavior by the leader node of the control chain, such as falsely claiming that certain shards are dead when they are still operational. Additionally, we do not address the potential scenario in which the control block is not accepted, which could introduce further complexities.

In terms of setting the time-bound for block generation and a round of voting within the shard, our approach closely follows the simulation design employed in the Gearbox paper. Within the Gearbox paper, it is mentioned that, in the context of the control chain, the latency in milliseconds for messages involving committees (comprising nodes) of size $s$ can be reasonably approximated using the linear functions $l = 0.37s + 6$. Similarly, for the shard, the approximation is $l = 0.67s + 20$. Unfortunately, the paper does not delve into the specifics of how these numerical values were derived.

In our paper, we also take into account the latency incurred when disseminating data, and we model this latency as linearly proportional to the number of nodes involved. To ensure the time-bound for shard-related activities is greater than the latency incurred during message transmission, we've established a relationship that satisfies $l = 2.12x$ where $x$ represents the number of nodes within the shard. Note that this estimation does not consider the potential slowdown when a node is in multiple shards and have multiple workload in parallel.  $2.12$ is chosen base on the fact that Rapidchain would be able to generate a block and conduct one round of voting that involves $329$ nodes in a shard within $698s$, which is measured in runtime in the same experiment environment for Reticulum and Rapidchain. Consequently, the time-bound for block generation and the subsequent voting rounds are set at $55.12s$, $82.68$, $106s$, $133.56s$, and $621.16s$, corresponding to the various shard sizes in different gears. This approach ensures that the time intervals are both consistent with the Gearbox paper's principles and suitably adjusted to account for the latency incurred in our specific context. Fig.~\ref{fig:GearboxThroughputfull} shows the performance of the comparison between Geaxbox and Reticulum in terms of throughput and Fig.~\ref{fig:GearboxStorage} shows that for the storage. The storage and download bandwidth comparison between Reticulum and Baseline are given in Fig.~\ref{fig:RapidchainStorage} and Fig.~\ref{fig:RapidchainBandwidth} respectively.

\begin{figure*}
    \centering
    \includegraphics[width=\textwidth]{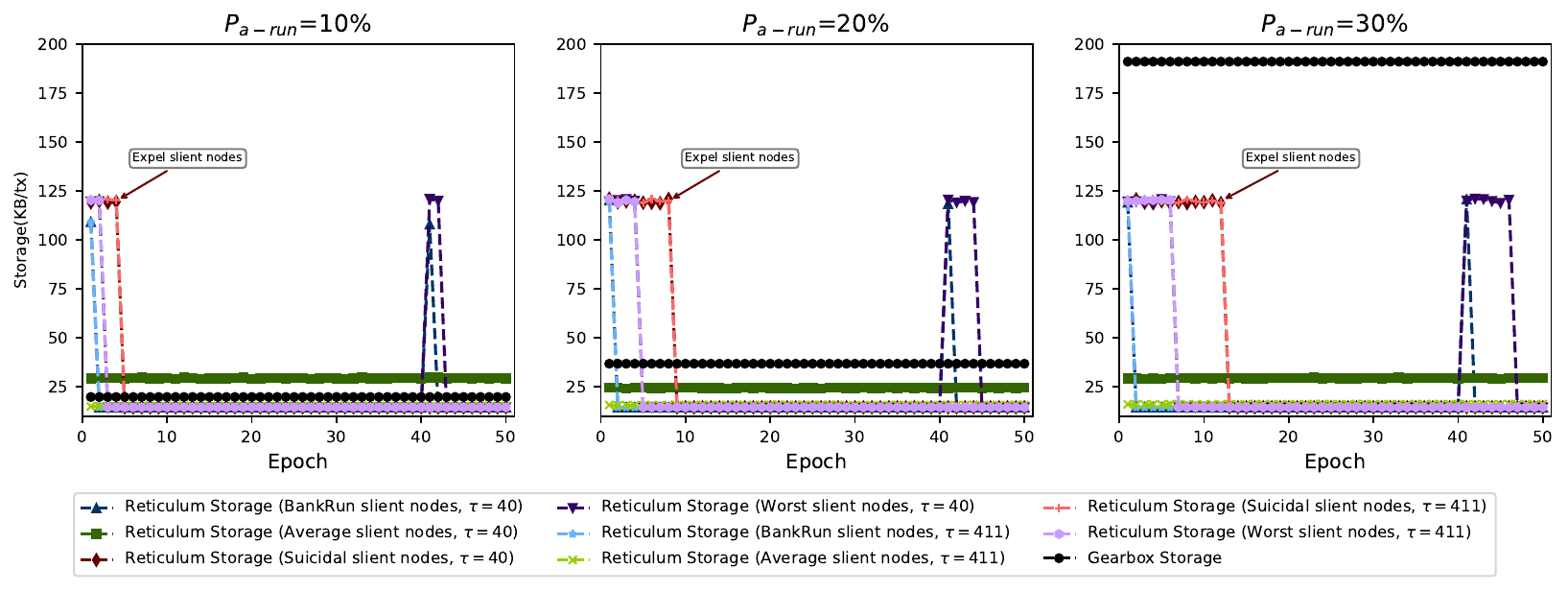}
    \caption{The experimental results in terms of storage occurred over the entire network of Reticulum with $(P_{f})_{threshold}=10^{-7}$, $N=5000$, $N_c=329$, $N_p=21$, $T_1=86s$, $\lambda=800$, $B_s=2MB$ (a block has 4096 transactions) and $\tau=40$. The figure also shows the result of Gearbox with $(P_{f})_{threshold}=10^{-7}$, $N=5000$ and $B_s=2MB$ (a block has 4096 transactions) after each shard has found a stable gear. As can be seen from the picture, Reticulum uncontestedly outperforms Gearbox in all cases. Note that the storage does not refer to the storage for a single node, but the storage occurred over the entire network.}
    \label{fig:GearboxStorage}
\end{figure*}

\begin{figure*}
    \centering
    \includegraphics[width=\textwidth]{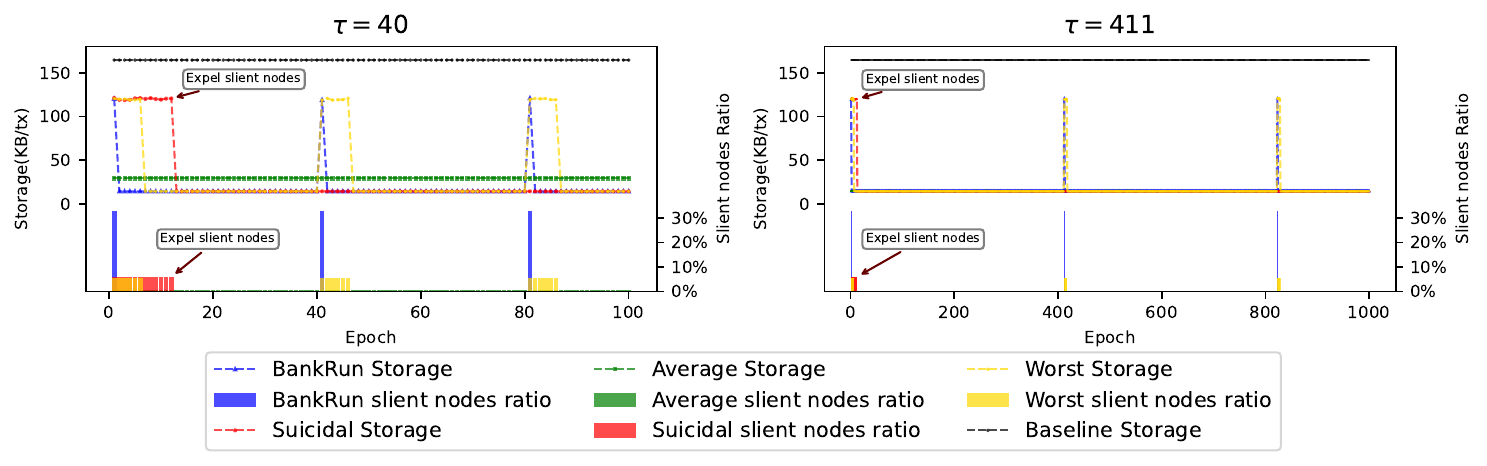}
    \caption{The experiment result (storage) of a system of Reticulum with $(P_{f})_{threshold}=10^{-7}$, $N=5000$, $N_c=329$, $N_p=21$, $P_a=33\%$, $T_1=86s$, $\lambda=800$, $B_s=2MB$ (a block has 4096 transactions) and different $\tau$ and the result of Baseline with $(P_{f})_{threshold}=10^{-7}$, $N=5000$ and $E_{time}=698s$ and each shard sized $329$. $B_s=2MB$ (a block has 4096 transactions). Silent nodes (globally not just within a shard) stand for the nodes that did not vote for the process shard at the given epoch. The experiment setup matches the simulation setup we did in Sec.~\ref{theoratical_rapidchain}. As can be seen from the picture, Reticulum outperforms Baseline in all cases. When $\tau$ is larger, the attack $worst$ occurs less frequently. We need at least $\tau$+1 epochs to illustrate BankRun attacks. Therefore, the experiment last $100$ epochs for $\tau=40$  and $1000$ epochs for $\tau=411$.  Because transactions are handled in the process shard level, which consists of fewer nodes than a shard in Baseline, nodes keep fewer transactions compared to Baseline. Note that the storage does not refer to the storage for a single node, but the storage occurred over the entire network.}
    \label{fig:RapidchainStorage}
\end{figure*}
\begin{figure*}
    \centering
    \includegraphics[width=\textwidth]{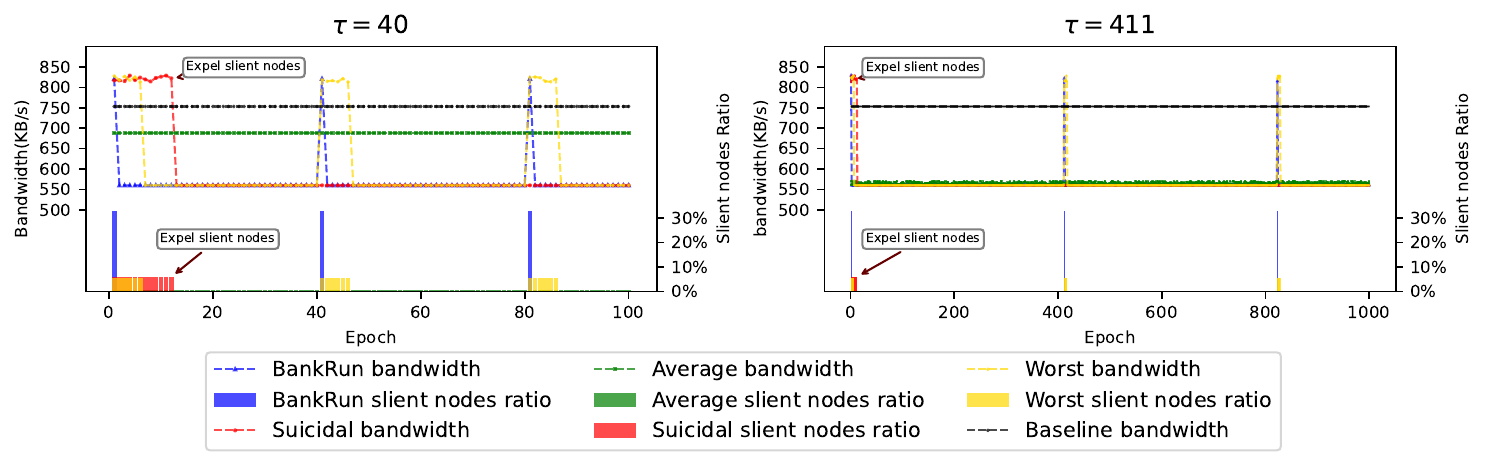}
    \caption{The experiment result (download bandwidth) of a system of Reticulum with $(P_{f})_{threshold}=10^{-7}$, $N=5000$, $N_c=329$, $N_p=21$, $P_a=33\%$, $T_1=86s$, $\lambda=800$, $B_s=2MB$ (a block has 4096 transactions) and different $\tau$ and the result of Baseline with $(P_{f})_{threshold}=10^{-7}$, $N=5000$ and $E_{time}=698s$ and each shard sized $329$. $B_s=2MB$ (a block has 4096 transactions). Silent nodes (globally not just within a shard) stand for the nodes that did not vote for the process shard at the given epoch.}% The experiment setup matches the simulation setup we did in Sec.~\ref{theoratical_rapidchain}. }
    \label{fig:RapidchainBandwidth}
\end{figure*}

\section{Related Work}
Several sharding protocols~\cite{luu2016secure,xu2020n,xu2023two,kokoris2018omniledger,zamani2018rapidchain,avarikioti2019divide} have been proposed in recent years to address scalability issues in blockchain systems. This section reviews some of the most representative ones. It should be noted that there are alternative approaches, such as Prism~\cite{bagaria2019prism} and OHIE~\cite{yu2020ohie}, which are limited to the Proof-of-Work (PoW) setting, selecting blocks via resource competition rather than a leader-and-vote-based consensus, and will not be discussed here.

A common issue of the works is that they only consider tolerating an upper bound of the adversarial population, which may not perform well when the adversarial population is lower than the upper bound. Recent efforts~\cite{xu2020flexible,david2022Gearbox} attempt to adjust the shards according to the adversarial population. The first work for sharding that uses $S>50\%$ to the knowledge of the authors, is proposed in 2020~\cite{xu2020n}, a proposal which seeks to achieve $f\leq (M-1)/2$ in a shard with $f\leq (N-1)/2$ in overall via node classification. However, the method's effectiveness depends on the adversary's strategy, which is inaccurately represented in the paper. Consequently, the protocol's reliability in real-world scenarios is uncertain. A protocol based on~\cite{xu2020n} enabling two-way adjustments to shard size based on the actual percentage of the adversarial population and the runtime workload, without overlapping shard membership, is proposed in~\cite{xu2020flexible}. However, the protocol is vulnerable to attacks by adversaries who frequently trigger adjustments in both directions, causing significant overhead for adjusting shards and data synchronization. Additionally, the paper inherits the error of~\cite{xu2020n}, which makes it unreliable. Gearbox~\cite{david2022Gearbox} is a protocol that enables one-directional adjustment of shard size based on the percentage of the adversarial population in the network. It operates under the assumption that the actual adversarial population in a shard is unknown but fixed and below the worst-case ($(L=S)<50\%$). Gearbox cannot determine the runtime adversary population; it can only approximate it by gradually ``switching gears.'' The system starts with a small shard size and hence a small $L$ value. If a shard loses liveness, Gearbox continuously increases the shard size and its $L$ value until it surpasses the actual adversarial population in the shards. At this point, the shard is considered alive and can output blocks. This approach, however, does not account for runtime fluctuations in the population of adversarial nodes or the presence of different adversaries over time. Maintaining a shard size that accommodates an all-time high adversarial population is inefficient when the percentage may vary at runtime. Simply making shard size adjustments in both directions to enable runtime increase and decrease of $L$ will not work in practice. Because the attacker can easily trigger a loop in ``switching gears'', resulting in frequent and costly shard adjustments. Additionally, overlapping shards can occur in Gearbox, because it only resizes shards without adjusting the number of shards. The overlapping shards can suffer from duplicated workloads, and more transactions need to involve more than one shard, resulting in more cross-shard transactions and less parallelism. To maintain non-overlapping shards, resizing them to form new ones when a single shard loses liveness is necessary. However, this needs to divide or merge transactions governed by the original shards into new ones, which is complicated and costly. Thus, despite the limitation, Gearbox opts for not adjusting the number of shards to avoid the aforementioned issues. Lastly, the previous two works~\cite{xu2020n,xu2020flexible} and Gearbox are insecure, as shown in Lem.~\ref{llemma1}.

\section{Conclusion}
In conclusion, the Reticulum sharding protocol presents a promising solution for enhancing the scalability of blockchain technology. By separately considering attacks on liveness and safety and providing methods for liveness attack inhibition, Reticulum provided superior performance. The two-phase design of Reticulum, consisting of control and process shards, enables nodes to detect adversarial behaviours and further split the workload and storage into tiny groups (process shards). The protocol leverages unanimous voting in the first phase to involve fewer nodes for accepting/rejecting a block, allowing more parallel process shards, while the control shard comes into play for consensus finalization and as a liveness rescue when disputes arise. Overall, the analysis and  experimental results demonstrate that Reticulum is a superior sharding protocol for blockchain networks.

{\bibliographystyle{acm}
\bibliography{bib}}

\appendix
\renewcommand{\thealgorithm}{\Alph{section}.\arabic{algorithm}}
\setcounter{algorithm}{2}
\subsection{Security analysis for the design of Reticulum protocol}
\subsubsection{Bootstrapping} \label{a1}

\textbf{Lemma.} \textit{
The bootstrapping phase guarantees the safety property of shard membership assignment, ensuring that no node is assigned to multiple shards.}
\begin{proof}
    To prove the safety property of the bootstrapping phase, we need to show that no node is assigned to multiple shards, i.e., there are no conflicts in shard membership assignments.

Suppose there exists a node $Node_j$ that is assigned to both $ps_{index_1}$ and $cs_{index_2}$, where $ps_{index_1}$ represents the process shard with index $index_1$, and $cs_{index_2}$ represents the control shard with index $index_2$. This would imply a conflict in shard membership assignment.

Let's consider the getShardIndex function and its steps:
\begin{enumerate}\small
    \item $C_{index}$ is the index of $Node_j$ in the sequence $C$. Since $C$ is generated using a random beacon and each node is assigned a unique index, there are no duplicate indices.
    \item $index_1$ is calculated as $\lfloor C_{index} / N_p \rfloor$. This calculation ensures that $Node_j$ is assigned to a specific process shard based on its index.
    \item $index_2$ is calculated as $\lfloor index_1 / (N_c/N_p) \rfloor$. This calculation ensures that $Node_j$ is assigned to a specific control shard based on its process shard index.
    \item If $C_{index}$ exceeds or is equal to $\lfloor N/N_c \rfloor \times N_c$, $index_2$ is decremented by 1 to avoid the creation of a new control shard. This adjustment ensures that nodes are not duplicated across multiple shards.
    \item If $N \nmid N_p$, additional nodes are added to the last process shard but not a new process shard. The calculation of $index_1$ limits the assignment of nodes to the last process shard, avoiding duplication.
\end{enumerate}

Based on the steps of the getShardIndex algorithm, we can conclude that the assignment of a node to a process shard and a control shard is unique and does not conflict with any other shard assignments. This guarantees the safety property of the bootstrapping phase, ensuring that no node is assigned to multiple shards.

Therefore, the bootstrapping phase design provides a safe and conflict-free mechanism for assigning shard memberships to nodes.

\end{proof}

\subsubsection{Two-layer Consensus Safety}
\label{a2}

\noindent\textbf{Theorem.}  The two-layer consensus protocol ensures safety by accepting blocks approved either unanimously in the first phase or by a majority vote in the second phase. Rejected blocks do not impact subsequent state evolution.

\begin{proof}
To prove safety:
\begin{enumerate}
\item Unanimously approved blocks in the first phase are considered safe and not processed further.
\item Blocks lacking unanimous approval proceed to the second phase for verification.
\item Only process blocks with majority approval in the second phase are accepted, ensuring safety.
\item Rejected blocks in the second phase have no impact on subsequent state evolution, maintaining safety.
\end{enumerate}

\textbf{Proof of Claim 1:} Unanimously accepted blocks are safe, as all nodes in the process shard agree within the time-bound $T_1$. No adversarial behavior guarantees unanimous approval.

\textbf{Proof of Claim 2:} Blocks lacking unanimous approval proceed to the second phase for verification, ensuring safety against potential adversarial nodes.

\textbf{Proof of Claim 3:} Majority approval in the second phase ensures safety by accepting process blocks marked as accepted by the leader node.

\textbf{Proof of Claim 4:} Rejected process blocks in an accepted control block do not alter subsequent state evolution, maintaining safety and consistency across epochs.

In conclusion, the two-layer consensus protocol provides a secure mechanism for distributed system consensus.
\end{proof}

\subsubsection{Cross-shard transactions} \label{a3}

\begin{algorithm}\label{a4}
    \caption{Handling of Cross-Shard Transactions}\label{cross-shard}
    \textbf{$ps_{receive}$}
\begin{algorithmic}[1]
    \Procedure{HandleTransaction}{$Tx_{cross}$}
        \If{signature is valid}
            \State Generate two transaction: $Tx_{cross1} \gets$ deduct from sender address; $Tx_{cross2} \gets$ deposit to recipient address;
        \Else
            \State Discard this transaction
            \State \Return
        \EndIf
        \State Send $Tx_{cross1}$ to $ps_{send}$
        \While{Get the Merkle proof of the process block that contains $Tx_{cross1}$ and the process block has been accepted}
        \State The current leader node adds $Tx_{cross2}$ to the process block of $PS_{receive}$ as a normal transaction. 
        \EndWhile
    \EndProcedure
\end{algorithmic}

\textbf{$ps_{send}$}
\begin{algorithmic}[1]
    \Procedure{Receive}{$Tx_{cross1}$}
    \If{signature is valid \& balance of sender is sufficient}
        \State The current leader node adds $Tx_{cross1}$ to the process block as a normal transaction
            \Else
            \State Discard this transaction
            \State \Return
        \EndIf
    \While{The block that contains $Tx_{cross1}$ has been accepted into the blockchain of $ps_{send}$}
        \State Send the Merkle proof of this process block to $ps_{receive}$
        \EndWhile
    \EndProcedure
\end{algorithmic}

\end{algorithm}

\noindent \textbf{Lemma.} \textit{The cross-shard transaction mechanism in Reticulum ensures the security and integrity of transactions between different process shards, providing transaction validity, consensus participation, proof of transaction, and a fixed shard membership approach.}

\begin{proof}
    
To prove the security of the cross-shard transaction mechanism in Reticulum, we examine each of the key aspects identified in the security analysis.

\textbf{Transaction Validity:}
The validity of each cross-shard transaction is ensured through signature verification and balance sufficiency checks. In Alg.~\ref{a4}, both the sender shard ($ps_{send}$) and recipient shard ($ps_{receive}$) validate the signature of the transaction and verify that the sender's balance is sufficient. Only transactions with valid signatures and sufficient balances are processed and included in the respective shards' blocks.

\textbf{Consensus Participation:}
Reticulum leverages a two-layer consensus process for including cross-shard transactions in the blocks of their respective shards. This ensures that the transactions are agreed upon by the participating nodes. The consensus process guarantees that the majority of nodes reach a consensus on the validity and order of transactions, preventing malicious actors from manipulating the transaction history.

\textbf{Proof of Transaction:}
The cross-shard transaction mechanism in Reticulum includes the Merkle proof of the block containing $Tx_{cross1}$ as proof to the recipient shard, $ps_{receive}$. This proof allows $ps_{receive}$ to verify the authenticity and integrity of the transaction. By validating the Merkle proof, $ps_{receive}$ can confirm that the sender has sufficient funds and that the funds have been deducted. This proof mechanism ensures the security and correctness of cross-shard transactions.

\textbf{Fixed Shard Membership:}
Reticulum's fixed shard membership approach simplifies the handling of cross-shard transactions for approaches that leverage liveness and security by avoiding the complexities of dead shards and shard membership changes. With a fixed set of nodes participating in the consensus process for an epoch, the same as the classical protocols like Rapidchain or Ominiledger, the system achieves stability and security. Reticulum's deadlock-free design much simplified the process and the time to finalize the cross-shard transactions.

Based on the above analysis, the cross-shard transaction mechanism in Reticulum ensures the security and integrity of transactions between different process shards. By enforcing transaction validity, leveraging consensus participation, providing proof of transaction, and employing a fixed node membership approach, Reticulum establishes a secure foundation for cross-shard transactions.

Therefore, we can conclude that the lemma holds, and the security of cross-shard transactions in Reticulum is assured.
\end{proof}
\subsection{Artifact Appendix}

\textit{This section shows the artifact description we submitted to the NDSS 2024 Artifact Evaluation committees. The committees evaluated a preliminary version of the artifact, granting the Functional and Reproduced badges. Later in time, the source materials of the artifact benefited from several enhancements and additions.}

This artifact appendix presents a detailed derivation process for figures ranging from Fig. 6 to Fig. 11 in the paper. While we aimed to provide comprehensive design details and open-source code for the Gearbox component, we had to resort to mathematical simulations to obtain results due to the unavailability of detailed design from the Gearbox paper. Conversely, we were able to gather results for Reticulum and Rapidchain through experimentation. We developed a protocol that replicate the relevant details of Rapidchain and denoted as Baseline. 

In the experiment section, we showcased the outcomes of a system consisting of 5000 nodes distributed across fifteen servers, each equipped with 32-core AMD EPYC 7R13 processors. These servers provided 128 virtual CPUs running at 3.30 GHz, 256 GB of memory, and a network bandwidth of 10 Gbps. However, it's important to note that the experimentation process incurred significant costs.

Regrettably, these experiments do not align with the artifact committee's requirements, which stipulate that proposed experiments should be executable in a single day and run on a standard desktop machine. As such, we've conducted a scaled-down experiment using only 45 nodes.

It's worth highlighting that the security of sharded blockchains heavily relies on a large number of nodes. Consequently, it was impossible to maintain the proposed $10^{-7}$ failure probability with just 45 nodes. \footnote{The failure proablity refers to the probablity for the adversarial nodes which are the minority in the system, accumulate enough nodes in a shard to manipulate the shard consensus.} In this reduced setup, we had to adjust the failure probability to $10^{-2}$. Under these conditions, we were able to safely implement one control shard comprising 45 nodes and three processing shards, each containing 15 nodes. Detailed calculations regarding the numerical relationship between the failure probability and the sizes of the control and processing shards can be found in Section IV.D of the paper.

It is essential to mention that experimenting with only one control shard is reasonable, as control shards operate independently and only communicate when cross-shard transactions occur. However, it's important to note that our paper does not include benchmarks for cross-shard transactions, and our experiments do not involve them either. Consequently, the transactions per second for the entire system can be easily calculated as a simple multiplication of the transaction rate within a single shard.

Furthermore, when considering storage requirements for nodes, we only take into account the storage that occurs within a control shard. Therefore, it is reasonable to perform testing with just one control shard.

\subsubsection{Downscaled Experiment for Rapidchain and Reticulum}
\begin{itemize}
    \item \textbf{Accessing the Resources.} We make these materials available to the community at \url{https://github.com/WaterandCola/Reticulum/}. We offered a cloud instance with an x86-64 CPU, 8 cores, and 16 GB of RAM (c5.2xlarge) during the artifact evaluation period.
    \item \textbf{Hardware Dependencies.} The experiment relies on a cloud instance with the following specifications: \\- x86-64 CPU\\ - 8 cores\\ - 16 GB of RAM
\item \textbf{Software Dependencies.} Ensure that the following software components are installed: \\- golang-1.19 \\- Ubuntu 20.04
\end{itemize}

\noindent \textbf{Explanation of Time-Bounds Settings.}
Both Reticulum and RapidChain operate as synchronous models. When a time-bound is reached, a decision is made based on voting, and the next epoch begins (unless the block is rejected). To maintain consistent bandwidth usage per node, we have set the following time-bound parameters:\\
- $T_1 = 16s$\\
- $\lambda = 9s$\\
- $T_{2_{max}} = (3 + 1) \times 9 = 36s$\\
- $\delta = 1s$

Using the calculation formula from Section V.D, we obtain:
$$
UB = \frac{2MB \times 45 + 30 \times 256KB}{T_{2_{max}} - \delta} = 2.78MB/s.
$$

The epoch time for Rapidchain is $E_{time} = 16.2s$, ensuring that both approaches require the same upload bandwidth.

\subsubsection{Simulation for Gearbox.}
\begin{itemize}
    \item \noindent\textbf{Accessing the Simulation.}
Please refer to the file \textit{Gearbox 0930.ipynb} in the Github page for access to the simulation. Note that because it is purely mathematical simulation, consistent to the original experiment in the paper, we show the situation for a system containing 5000 nodes.
\item \noindent \textbf{Hardware Dependencies.}
No specific hardware dependencies are required for this simulation.
\item \noindent \textbf{Software Dependencies}
The simulation does not have any specific software dependencies.
\end{itemize}

\subsubsection{Benchmarks}
We have conducted two benchmarks, measuring transactions per second and storage while considering various attack types for Reticulum, as outlined in the Attack Strategies section of Section VI. We set $\tau = 40$ (as defined in the Threat Model of Section III.A). The benchmarks are defined in the Threat Model of Section III.A  as follows:
- Transaction per second is referred to as ``Throughput.''
- Storage is referred to as ``Storage.''

\subsubsection{Artifact Installation \& Configuration}

Please refer to our GitHub page. If you are using our cloud instance, then run \url{./runall.sh} directly. After the running of the experiment, run the following terminal command to generate the experiment result: \textit{jupyter-nbconvert --to=html --ExecutePreprocessor.enabled=True Figure\ of\ Experiment.ipynb}

\subsubsection{Experiment Workflow}

The experiment follows the following workflow:

\begin{itemize}
    \item  [1.] Each node in the network queries the Drand API to obtain the current random number, determining their placement in process and control shards.
    \item  [2.] Leader nodes for the process and control shards are randomly selected at the beginning of the experiment.
    \item  [3.] Blockchain epochs commence when a control block from the previous blockchain is accepted, signifying the end of the previous epoch.
    \item  [4.] During a blockchain epoch, each process shard generates and votes for process blocks.
    \item  [5.] After the time-bound $T_1$ elapses, the control shard collects process blocks that didn't achieve unanimous voting within $T_1$ and generates a control block.
    \item  [6.] All members within the control shard cast their votes for the control block.
    \item  [7.] The experiment concludes after a pre-defined number of epochs. The benchmark results are presented graphically, akin to those in the research paper.
\end{itemize}

\subsubsection{Major Claims}

\begin{itemize}
    \item \textbf{\textit{Claim 1:}} We propose Reticulum, a pioneering protocol that achieves both liveness and safety while mitigating security vulnerabilities. Reticulum effectively prevents adversarial behaviors, all without the need for dynamic shard respawning or overlapping shard memberships, thus minimizing unnecessary overhead. This claim is substantiated by our design details and proofs, independent of the experiment.
    
    \item \textbf{\textit{Claim 2:}} We conduct a comprehensive performance analysis of Reticulum, comparing it with state-of-the-art approaches. Our empirical evaluation encompasses both simulations and experiments, wherein Reticulum outperforms two leading protocols, Gearbox and Rapidchain, in terms of transaction throughput and storage requirements. We aim to substantiate this claim through the artifact evaluation.
    
    \item \textbf{\textit{Claim 3:}} We have developed and openly shared a fully functional prototype of Reticulum. This implementation addresses the dearth of available block sharding protocol implementations, facilitating comparisons between different protocols through experimentation. Researchers can utilize our open-source prototype as a valuable resource for blockchain sharding research.
\end{itemize}

\subsubsection{Evaluation}
\begin{itemize}
    \item [1.] Download the project code [1 human-minute + 3 compute-minutes].
    \item [2.] Run the project code [1 human-minute + 3200 compute-minutes]. The experiment spans 100 epochs.
    \item [3.] Generate a comparison of results between Rapidchain and Reticulum [1 human-minute + 1 compute-minute].
\end{itemize}

\noindent \textbf{[Results]}
The results showcase benchmarks for one control shard of Reticulum and one shard of Rapidchain. Since we don't test cross-shard transactions, the number of shards determines overall transaction throughput. With an increasing number of nodes, more shards become available. Currently, we have one control shard and three process shards for Reticulum, and only one shard for Rapidchain. When we scale up to 5000 nodes, the test results for both approaches will see significant improvements. However, Rapidchain's performance remains significantly lower compared to Reticulum.
\end{document}